\documentclass[reqno,11pt]{amsart}

\IfFileExists{mymtpro2.sty}{%
  \usepackage[subscriptcorrection]{mymtpro2}
}{}

\usepackage{a4,amssymb}
\usepackage{graphicx}

\newtheorem{theorem}{Theorem}[section]
\newtheorem{lemma}{Lemma}[section]
\newtheorem{follow}{Corollary}[section]
\newtheorem{pr}{Proposition}[section]
\theoremstyle{definition}
\newtheorem{remark}[theorem]{Remark}

\newcommand{\labelnummer}{\mbox{\normalfont (\roman{numcount})}}%

\makeatletter

  {\let\curlabelspeicher\@currentlabel%
    \begin{list}{\labelnummer}%
      {\usecounter{numcount}\leftmargin0pt%
        \topsep0.5ex\partopsep2ex\parsep0pt\itemsep0ex\@plus1\p@%
        \labelwidth2.5em\itemindent3.5em\labelsep1em%
      }%
    \let\saveitem\item%
    \def\item{\saveitem%
      \def\@currentlabel{{\upshape\curlabelspeicher}$\,$\labelnummer}}%
    \let\savelabel\label%
    \def\label##1{\savelabel{##1}%
      \@bsphack%
        \ifmmode\else%
          \protected@write\@auxout{}%
          {\string\newlabel{##1item}{{\labelnummer}{\thepage}}}%
        \fi%
      \@esphack%
    }%
  }{\end{list}}%

\renewcommand{\appendix}{\def\thesection{\textsc{Appendix}}}

 \let\leq\le
 \let\geq\ge

\DeclareMathOperator{\tr}{tr\kern1pt}

\makeatletter

\newif\ifper\pertrue
\def\per{.}

\def\bti{\@ifnextchar[\bbti\bbbti}
\def\bbti[#1]#2{#2, #1.}
\def\bbbti#1{#1.}

\def\z{\@ifnextchar[\zz\zzz}
\def\zz[#1]#2#3#4#5{\perfalse\emph{#2} \textbf{#3}, #4 (#5) [#1]}
\def\zzz#1#2#3#4{\emph{#1} \textbf{#2}, #3 (#4)\ifper\per\fi\pertrue}

\def\pub{\@ifstar\pubstar\pubnostar}
\def\pubnostar{\@ifnextchar[\@@pubnostar\@pubnostar}
\def\@@pubnostar[#1]#2#3#4{#2, #3, #4, #1\ifper\per\fi\pertrue}
\def\@pubnostar#1#2#3{#1, #2, #3\ifper\per\fi\pertrue}
\def\pubstar[#1]#2#3#4{\perfalse #2, #3, #4 [#1]\pertrue}

\makeatother

\newcommand{\bel}{\begin{equation} \label}
\newcommand{\ee}{\end{equation}}

\newcommand{\eps}{\epsilon}

\newcommand{\re}{{\mathbb R}}

\def\beq{\begin{equation}}
\def\eeq{\end{equation}}
\newcommand{\bea}{\begin{eqnarray}}
\newcommand{\eea}{\end{eqnarray}}
\newcommand{\beas}{\begin{eqnarray*}}
\newcommand{\eeas}{\end{eqnarray*}}
\newcommand{\Pre}[1]{\ensuremath{\mathrm{Re} \left( #1 \right)}}

{

\newcommand{\R}{\mathbb{R}}

\newcommand{\N}{\mathbb{N}}

\begin{document}

\title[Edge states for Iwatsuka Hamiltonians]{Edge states induced by Iwatsuka Hamiltonians with positive magnetic fields}

\author[P.\ D.\ Hislop]{Peter D.\ Hislop}
\address{Department of Mathematics,
    University of Kentucky,
    Lexington, Kentucky  40506-0027, USA}
\email{hislop@ms.uky.edu}

\author[E.\ Soccorsi]{Eric Soccorsi}
\address{Aix Marseille Universit\'e,
CNRS, CPT, UMR 7332, 13288 Marseille, France
\& Universit\'e du Sud-Toulon-Var, CNRS, CPT, UMR 7332, 83957 La Garde, France }
\email{eric.soccorsi@univ-amu.fr}

\thanks{Version of \today}

\begin{abstract}
We study purely magnetic Schr\"odinger operators in two-dimensions $(x,y)$ with magnetic fields $b(x)$ that depend
only on the $x$-coordinate. The magnetic field $b(x)$ is assumed to be bounded, there are constants $0 < b_- < b_+ < \infty$
so that $b_- \leq b(x) \leq b_+$, and outside of a strip of
small width $-\epsilon < x < \epsilon$, where $0 < \epsilon < b_-^{-1/2}$, we have $b(x) = b_\pm x$ for $\pm x > \epsilon$.
The case of a jump in the magnetic field at $x=0$ corresponding to $\epsilon=0$ is also studied.
We prove that the magnetic field creates an effective barrier near $x=0$
that causes edge currents to flow along it consistent with the classical interpretation.
We prove lower bounds on edge currents carried by states with energy localized inside the energy bands of the Hamiltonian.
We prove that these edge current-carrying states are well-localized in $x$ to a region of size $b_-^{-1/2}$,
also consistent with the classical interpretation.
We demonstrate that the edge currents are stable with respect to various magnetic and electric perturbations. These lower bounds on the edge current hold for all time.
For a family of perturbations compactly supported in the $y$-direction, we prove that the time asymptotic current exists and satisfies the same lower bound.
\end{abstract}

\maketitle \thispagestyle{empty}

\tableofcontents

\vspace{.2in}

\noindent {\bf  AMS 2000 Mathematics Subject Classification:} 35J10, 81Q10,
35P20.\\
\noindent {\bf  Keywords:}
Schr\"odinger operators, magnetic field, magnetic edge states, magnetic barrier, asymptotic velocity. \\


\section{Introduction: Magnetic barriers and edge currents}\label{sec:intro1}
\setcounter{equation}{0}

Quantum Hall systems describe charge transport in bounded or unbounded regions in the plane in the presence of a transverse magnetic field. The typical Hall system is described by an electron moving in the plane subject to a {\em constant} transverse magnetic field. The Hall conductance is quantized \cite{bvesb,combes-germinet} and is stable under perturbations by random potentials. Confined systems, such as motion in a half-plane or a strip are also interesting as a current flowing along an edge is created. Confinement may be obtained by Dirichlet boundary conditions or an electrostatic potential barrier. The edge currents in these situations were explored in \cite{CHS,fgw,HS1,HS2}. In this article, we are interested in edge currents created by purely magnetic barriers.

The spectral properties of magnetic Schr\"odinger operators of the form $H = (-i \nabla - A)^2$ on ${\rm L}^2 (\R^2)$ have been of interest for many years. It is known that they depend only on the transverse magnetic field given by the third component of the cross product:
$b(x,y) = ( \nabla \times A(x,y) )_3$.
The case when $b(x,y) = b_0$ is the Landau Hamiltonian.
The spectrum is pure point consisting of infinitely-degenerate eigenvalues $E_n(b_0) = (2n + 1)b_0$, for $n=0,1,2, \ldots$, the Landau levels. Many works consider the case where the magnetic field is asymptotically constant $b(x,y) \rightarrow b_0$ as $r = \sqrt{x^2 + y^2} \rightarrow \infty$. If $b_0 = 0$, then the essential spectrum is the half-line $[0, \infty)$ (see, for example, \cite{leinfelder}). If, in addition, the magnetic field is short-range $|b(x,y)| \sim |(x,y)|^{-1 - \delta}$, for any $\delta > 0$, then the spectrum is purely absolutely continuous
(see, for example, \cite{ikebe-saito}).



Iwatsuka \cite{iwatsuka1} studied the case
when the magnetic field is not constant at infinity. In particular, he considered the case when $b(x,y) = b(x)$ is a
function of $x$ only. He supposed that $b(x)$ is bounded, $0 < M_- \leq b(x) \leq M_+ < \infty$. We are interested in his
model for which $b(x)$ has different limits as
$x \rightarrow \pm \infty$. Under these conditions, Iwatsuka proved that the spectrum is absolutely continuous.
Several years later, the transport properties of purely magnetic Schr\"odinger operators on ${\rm L}^2 (\R^2)$ were investigated by
physicists Reijniers and Peeters \cite{reijniers-peeters}.
They supposed that $b(x)$ assumes constant value $b_-$ for $x<0$ and $b_+$ for $x>0$.
They argued that this magnetic discontinuity creates an effective edge and that currents flow along the edge.
This is the magnetic analog of the barriers created by Dirichlet bound conditions along $x=0$ or a confining electrostatic
potential filling the half-space $x<0$ described in the paragraph above.
Partially motivated by \cite{reijniers-peeters}, Dombrowski, Germinet, and Raikov \cite{DGR} studied the edge conductance for generalized Iwatsuka models. We discuss this in section \ref{subsec:relation1}.

In this article, we prove the existence, localization, and stability of magnetic edge currents.
We consider a family of magnetic fields $b(x)$ with the following properties. We first consider an Iwatsuka-type model with a sharp transition at $x = 0$. Let $0<b_-<b_+<\infty$ and $b(x):= b_{\pm}$ for $x \in \re_{\pm}^* := \re_\pm \backslash \{ 0 \}$.
Following the notation of \cite{DGR}, we set
\bel{a1}
\beta(x):=\int_0^x b(s) ds = b_{\pm} x,\ x \in \re_{\pm}^*,
\ee
and consider the two-dimensional vector potential $A(x,y) = (A_1 (x,y), A_2 (x,y))$ defined by $A_1 (x,y) :=0$ and $A_2 (x,y):=\beta(x)$ generating this magnetic field.

Let $p_x := -i \partial_x$ and $p_y := - i \partial_y$ be the two
momentum operators. The two-dimensional magnetic Schr\"odinger operator $H(A)$
is defined on the dense domain ${\rm C}_0^\infty (\re^2) \subset {\rm L}^2 (\re^2)$ by
\bel{eq:sharp1}
H=H(A) := (-{\rm i} \nabla - A)^2 = p_x^2 + (p_y - \beta(x))^2.
\ee
We will call the Hamiltonian in \eqref{eq:sharp1} {\it the sharp Iwatsuka model}. We study the edge current flowing along $x=0$ created by the discontinuity in the magnetic field there. We prove the existence of states carrying a nonzero edge current. We show that the current is localized in a neighborhood about $x=0$ with width the order of $b_-^{-1/2}$.
We consider a smoothed version of $b(x)$ in section \ref{sec:smooth1} for which the magnetic field is bounded $0 < b_- \leq b(x) \leq b_+ < \infty$,
with $b(x) = b_\pm$ for $\pm x > \epsilon > 0$, for some $\epsilon > 0$. In order to preserve the localization of edge currents, we take $\epsilon < {b_-}^{-1/2}$.
Finally, we prove the currents are stable with respect to various families of magnetic and electric perturbations.

In a companion article with N.\ Dombrowski \cite{DHS}, we study the Iwatsuka model
for which $b_- = - b$, and $b_+ = b > 0$. The Hamiltonian for this model is symmetric with respect to the
reflection $x \rightarrow -x$ and the band functions have different asymptotics as $k \rightarrow \pm \infty$.
A characteristic of the model is the existence of so-called `snake orbits' (see \cite{reijniers-peeters}) along the magnetic edge $x=0$.

\subsection{Relation to edge conductance}\label{subsec:relation1}

As mentioned above, Dombrowski, Germinet, and Raikov \cite{DGR} studied the quantization of the Hall edge conductance for a generalized family of Iwatsuka models including the model discussed here.
Let us recall that the Hall edge conductance is defined as follows.
We consider the situation where the edge lies along the $y$-axis as discussed above.
Let $I \subset \R$ be a compact energy interval. We choose a smooth increasing function $g$ so $g(s) = s$ on $I = [a,b]$ and $0 \leq g \leq 1$. It follows that $g'_{ | [ a,b]} = 1$. We can arrange it so there is an $\sigma > 0$ so that ${\rm supp} ~g' \subset [a - \sigma, b + \sigma]$.
Let $\chi = \chi (y)$ be an $x$-translation invariant smooth function with $\mbox{supp} ~\chi' \subset [-1/2, 1/2]$. The edge Hall conductance is
defined by
\beq\label{eq:edgy1}
\sigma_e^I (H) = - 2 \pi {\rm tr} ~ ( g'(H) i[ H, \chi ] ) ,
\eeq
whenever it exists. The edge conductance measures the current across the axis $y=0$ with energies below the energy interval $I$.

One of the main results of \cite{DGR} in this setting is the quantization of edge currents for the Iwatsuka model.
Roughly speaking, for a fixed energy level $E$, the edge conductance at $E$ counts the number of Landau levels below $E$ carrying an edge current. Applied to the model studied here, for which $b(x) \rightarrow b_\pm$ as $x \rightarrow \pm \infty$, with $0 < b_- < b_+$, they proved \cite[Corollary 2.3]{DGR} for $I \subset \R$, any energy interval $I \subset (- \infty, b_+) \cap ( (2 n - 1) b_-, (2n + 1)b_-) \neq \emptyset$, for some positive integer $n \geq 1$, that the edge conductance is quantized: $\sigma_e^I (H) = n$. We complement this result as follows. We prove that there are $n$ nonempty intervals $\Delta_j, j=1, \ldots, n$ located below $I$ and a finite constant $c > 0$, see Theorem \ref{lm-a3}, so that for any state $\psi = \mathbb{P}(\Delta_j) \psi$, where $\mathbb{P}(\Delta_j)$ is the spectral projector for $H$ and interval $\Delta_j$,
we have
\beq\label{eq:edgy2}
\langle \psi, v_y \psi \rangle \geq c b_-^{1/2} \| \psi \|^2 > 0, ~~ v_y := p_y - \beta (x) .
\eeq
This indicates that such a state $\psi$ carries a nontrivial edge current.

\subsection{Contents}\label{subsec:contents1}

We recall the basic characteristics of the sharp Iwatsuka model in section 2 and present the standard fiber decomposition. The band functions are studied extensively in section 3. The main result, Theorem \ref{prop-a2}, provides a quantitative lower bound on the derivative of any band function for quasi-momentum $k$ in specified intervals. Edge currents and their spatial localization around $x=0$ are studied in section 4. In section 5, we treat the case of a smoothed magnetic field that is piecewise
constant outside of an interval $[- \epsilon, \epsilon]$, for any $0 < \epsilon < C_0 b_-^{-1/2}$.
We prove the existence and localization of edge currents in this case also. In section 6, we prove the stability of edge currents under certain families of magnetic and electric perturbations. We discuss lower bounds on the time-evolved edge current and show that they are stable under time evolution in section 7. In addition, we prove that for a class of perturbations that have compact support in the $y$-direction, the asymptotic velocity exists and is bounded below indicating the existence of edge currents for all time.


\subsection{Notation}\label{subsec:note1}

We write $\langle \cdot, \cdot \rangle$ and $\| \cdot \|$ for the inner product and norm on ${\rm L}^2(\R^2)$.
The functions are written with coordinates $(x,y)$, or, after a partial Fourier transform with respect to $y$,
we work with functions $f(x,k) \in L^2(\R^2)$. We often view these functions $f(x,k)$ on ${\rm L}^2 (\R_x)$ as parameterized by $k \in \R$.
In this case, we also write $\langle f(\cdot, k), g(\cdot, k ) \rangle$ and $\| f(\cdot, k) \|$ for the inner product
and related norm on $L^2 (\R_x)$. So whenever an explicit dependance on the parameter $k$ appears, the functions should be considered on $L^2 (\R_x)$. We indicate explicitly in the notation, such as $\| \cdot \|_X$, for $X = {\rm L}^2 (\R_\pm)$, when we work on those spaces. We write $\| \cdot \|_\infty$ for $\| \cdot \|_{{\rm L}^\infty (X)}$ for $X= \R, \R_\pm, ~{\rm or} ~ \R^2$.
For a subset $X \subset \R$, we denote by $X^*$ the set $X^* := X \backslash \{ 0 \}$.

\subsection{Acknowledgements}\label{subsec:acknowledgements1}

PDH thanks the Centre de Physique Th\'eorique, CNRS, Luminy, Marseille, France, for its hospitality. PDH was partially supported by the Universit\'e du Sud Toulon-Var, La Garde, France, and National Science Foundation grant 11-03104
during the time part of the work was done. ES thanks the University of Kentucky, Lexington, KY, USA, where part of this work was done, for its warm welcome.


\section{Preliminary analysis of the sharp Iwatsuka model with two constant magnetic fields}\label{sec:sharp1}

\setcounter{equation}{0}

Since the Hamiltonian defined in \eqref{a1}--\eqref{eq:sharp1} is invariant with respect to translations in the $y$-direction, it can be reduced to a family of parameterized Schr\"odinger operators on $L^2 (\R)$.
Let ${\mathcal F}$ denote the
partial Fourier transform with respect to $y$,
\beq\label{eq:fourier1}
 ({\mathcal F}u)(x,k): = \hat{u}(x,k) = \frac{1}{\sqrt{2\pi}} \int_{\re} e^{-{\rm i} yk}
u(x,y)dy,\ (x,k) \in \re^2.
\eeq
The operator $H$ admits a partial Fourier decomposition with
respect to the $y$-variable, and the Hilbert space ${\rm L}^2 (\re^2 )$ can be
expressed as a constant fiber direct integral over $\re$ with fibers ${\rm L}^2(\re)$,
\bel{a4}
{\mathcal F} H {\mathcal F}^*= \int^{\oplus}_{\re} h(k) dk
\ee
with
\bel{a5}
h(k) := p_x^2 + V(x,k)\ {\rm on}\ {\rm L}^2 (\re),\ V(x,k):=(k-\beta(x))^2.
\ee

In light of \eqref{a1}, the potential $V(x,k)$, $k \in \re$, is unbounded as $|x|$ goes to infinity, hence $h(k)$ has a compact resolvent.
Let $\left\{\omega_j(k)\right\}_{j=1}^{\infty}$ be the increasing
sequence of the eigenvalues of the operator $h(k)$, $k \in \re$.
Since all the eigenvalues $\omega_j(k)$ are simple (see \cite{HS1}[Proposition A2]), they depend
analytically on $k \in \re$ (see \cite{K} or \cite{RS}).
We refer to the functions $\omega_j(k)$ as the {\em band functions}.

Let us introduce two Landau Hamiltonians on $\re^2$ each with a constant magnetic field $b_-$ or $b_+$:
\bel{a8}
h_{\pm}(k):=p_x^2 + V_{\pm}(x,k),\ V_{\pm}(x,k):=(k-b_{\pm} x)^2.
\ee
We then have the simple comparison for the operators
\beq\label{eq:ev-bounds1}
h_-(k) \leq h(k) \leq h_+(k  b_+ \slash b_-), ~~ \mbox{for} ~~k \leq 0 ,
\eeq
and
\beq\label{eq:ev-bounds2}
h_-(k b_- \slash b_+) \leq h(k) \leq h_+(k), ~~ \mbox{for} ~~ k \geq 0.
\eeq
These inequalities and the mini-max principle imply that
\bel{a9}
(2j-1) b_- \leq \omega_j(k) \leq (2j-1)b_+,\ j \in {\mathbb N}^*.
\ee

Let $\left\{\psi_j(k)\right\}_{j=1}^{\infty}$ be the ${\rm L}^2(\re)$-normalized eigenfunctions of
$h(k)$ satisfying
\bel{a10}
h(k) \psi_j(x,k) = \omega_j(k) \psi_j(x,k),\ x \in \re.
\ee
We choose all $\psi_j(k)$ to be real and the ground state eigenfunction to satisfy $\psi_1(x,k) > 0$, for $x \in
\re$ and $k \in \re$. Since $V(.,k) \in {\rm C}^0(\re) \cap {\rm C}^{\infty}(\re^*)$, the functions $\psi_j(.,k) \in {\rm C}^2(\re) \cap {\rm C}^{\infty}(\re^*)$, $j \in {\mathbb N}^*$, from \cite{HS1}[Proposition A1].
Moreover, the orthogonal projections $P_j (k) := |\psi_j(k)\rangle \langle
\psi_j(k)|$ , $j \in {\mathbb N}^*$, depend analytically on $k$ (see \cite{K} or
\cite{RS}).\\

Similarly, we write $\{ \psi_j^{\pm}(k) \}_{j=1}^{+\infty}$ the ${\rm L}^2(\re)$-normalized and real analytic eigenfunctions of the
Landau Hamiltonians $h_{\pm}(k)$, with $\psi_1^{\pm}(x,k)>0$, for $x \in \re$ and $k \in \re$.


\section{Estimates on the derivative of the band functions}\label{sec:band-deriv1}
\setcounter{equation}{0}

We first derive two useful expressions for the derivative of the band functions.
Fix $j \in {\mathbb N}^*$. According to the Feynman-Hellmann Theorem, we have
\bel{eq:f-h1}
\omega_j'(k) =  \int_{\re} \left( \frac{d h}{dk}\right) (k) \psi_j(x,k) \psi_j(x,k) ~dx
 = 2 \int_{\re} (k-\beta(x)) \psi_j(x,k)^2 ~dx,
 \ee
since $\frac{d}{dk} h(k)=\frac{\partial}{\partial k} V(x,k)= 2 (k-\beta(x))$.
Hence, it follows that
\bea
\omega_j'(k) & = & -\sum_{\zeta=+,-} b_{\zeta} ^{-1} \int_{\re_{\zeta}} \frac{\partial }{\partial x} (k-b_{\zeta}x)^2 \psi_j(x,k)^2 dx \nonumber \\
& = & \sum_{\zeta=+,-} b_{\zeta}^{-1} \left( \zeta k^2 \psi_j(0,k)^2 + 2\int_{\re_{\zeta}} (k-b_{\zeta}x)^2 \psi_j(x,k) \psi_j'(x,k) dx \right), \label{a11}
\eea
by integrating by parts.
We used the fact that $\lim_{x \rightarrow \pm \infty} | V(x,k) | \psi_j(x,k)^2=0$. This follows from the decay of the eigenfunctions established in the proof of Theorem \ref{lm-a4}.
Putting \eqref{a10} and \eqref{a11} together we get that
\bea
\omega_j'(k) & = & -(b_-^{-1}-b_+^{-1}) k^2 \psi_j(0,k)^2 \nonumber \\
& & + 2 \sum_{\zeta=+,-} b_{\zeta} ^{-1} \int_{\re_{\zeta}} (\omega_j(k) \psi_j(x,k) + \psi_j''(x,k)) \psi_j'(x,k) dx \nonumber \\
& = & (b_-^{-1}-b_+^{-1}) \left( (\omega_j(k)-k^2) \psi_j(0,k)^2 + \psi_j'(0,k)^2 \right). \label{a12}
\eea


\subsection{Positivity of the derivative of the band functions}\label{subsec:positivity1}

We prove that the bands are monotone increasing functions of $k$.
Since
\bea\label{eq:identity1}
\omega_j(k) \psi_j(0,k)^2 &= & 2 \omega_j(k) \int_{-\infty}^0 \psi_j(x,k) \psi_j'(x,k) dx  \nonumber \\
 &=& 2 \int_{-\infty}^0 h(k) \psi_j(x,k) \psi_j'(x,k) dx,
 \eea
 we see that
\beas
\omega_j(k) \psi_j(0,k)^2 & = & 2 \int_{-\infty}^0 ( - \psi_j''(x,k) + (b_-x-k)^2 \psi_j(x,k) ) \psi_j'(x,k) )~ dx \nonumber \\
& = & -\psi_j'(0,k)^2 + \int_{-\infty}^0 (b_-x-k)^2 (\psi_j(x,k)^2)'  dx \nonumber \\
& = &  -\psi_j'(0,k)^2 - 2b_- \int_{-\infty}^0 (b_-x-k) \psi_j(x,k)^2  dx + k^2 \psi_j(0,k)^2,
\eeas
by standard computations and provided $\lim_{x \rightarrow - \infty} (b_- x - k)^2 \psi_j(x,k)^2 = 0$. As a consequence, we have
$$ (\omega_j(k)-k^2) \psi_j(0,k)^2 + \psi_j'(0,k)^2 = -2b_- \int_{-\infty}^0 (b_-x-k) \psi_j(x,k)^2  dx. $$
Substituting this into the right side of \eqref{a12}, we obtain
\bel{a12b}
\omega_j'(k) = 2 \left( \frac{b_+}{b_-} -1 \right)  \int_{-\infty}^0 (k-b_-x) \psi_j(x,k)^2  dx .
\ee
Note that for $k \geq 0$, the right side of \eqref{a12b} is positive.

In order to get an expression that is positive
for $k < 0$, we start in a similar manner
from the identity
$$\omega_j(k) \psi_j(0,k)^2 =-2 \omega_j(k) \int_0^{+\infty} \psi_j(x,k) \psi_j'(x,k) dx = -2 \int_0^{+\infty} h(k) \psi_j(x,k) \psi_j'(x,k) dx,$$
and we obtain in the same way that
$$ (\omega_j(k)-k^2) \psi_j(0,k)^2 + \psi_j'(0,k)^2 = 2b_+ \int_0^{+\infty} (b_+x-k) \psi_j(x,k)^2  dx, $$
and thus
\bel{a12c}
\omega_j'(k) = 2 \left( \frac{b_+}{b_-}-1 \right)  \int_0^{+\infty} (b_+ x-k) \psi_j(x,k)^2  dx,
\ee
with the aid of \eqref{a12}. Note that for $k < 0$, the right side of \eqref{a12c} is positive.

\noindent
Combining these two results, we obtain the following lemma.

\begin{lemma}
\label{lm-a1}
For every $j \in {\mathbb N}^*$ and for every $k \in \re$, the derivative of the band function is positive: $\omega_j'(k) > 0$.
\end{lemma}

\begin{figure}
\centering
\includegraphics{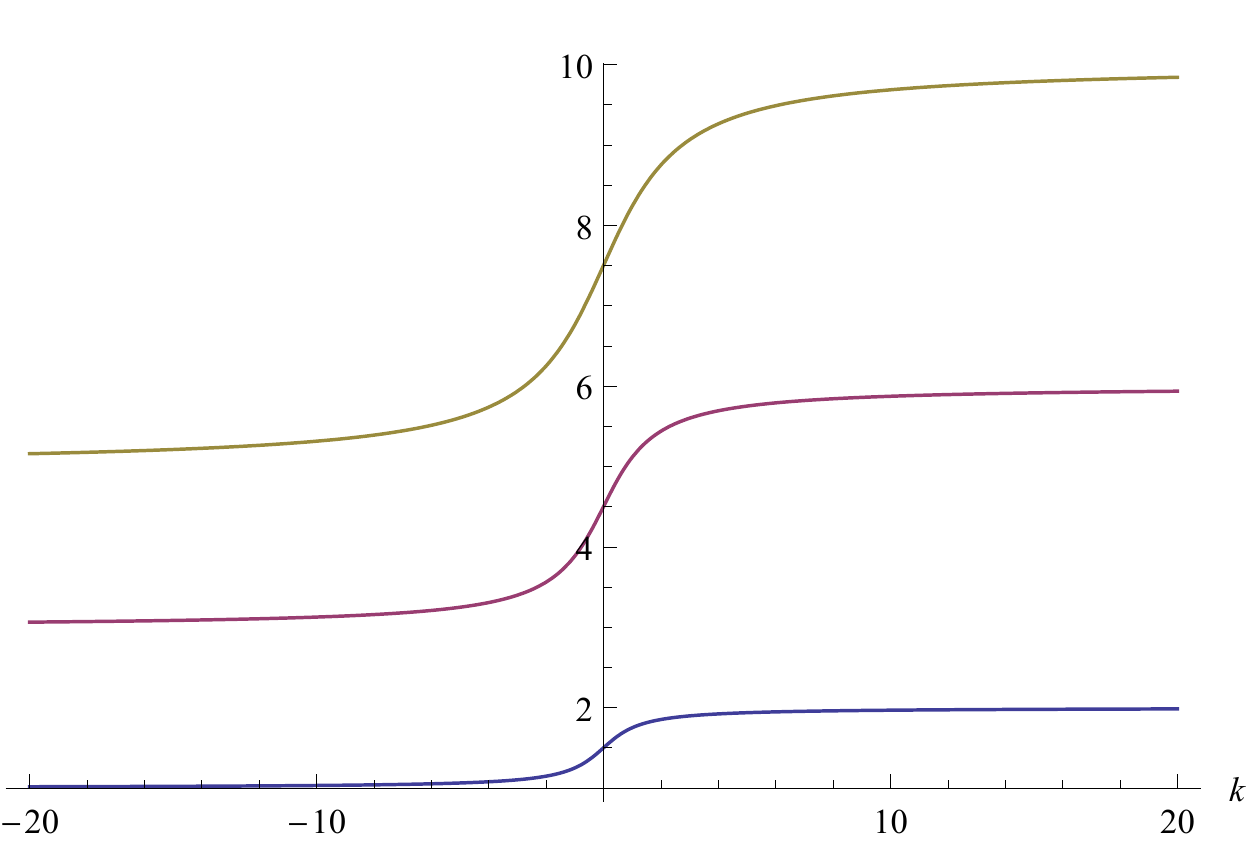}
\caption{Approximated shape of the band functions $k \mapsto \omega_j(k)$, for $j=1,2,3$, of the sharp Iwatsuka Hamiltonian with $b_+=2b_-=2$.}
\end{figure}

As the band functions $k \mapsto \omega_j(k)$, $j \in \N^*$, are non constant from Lemma \ref{lm-a1}, the spectrum of $H$ is purely absolutely continuous according to \cite{RS}[Theorem XIII.86]. Moreover we have
$\lim_{k \rightarrow \pm \infty} \| (h(k) - (2j-1) b_{\pm}) \psi_j^{\pm}(.,k) \| = 0$, directly from \eqref{a19}. From this and \eqref{a9} then follows that
\beq\label{eq:bflimit1}
\lim_{k \rightarrow \pm \infty} \omega_j(k) = (2j-1) b_\pm,\ j \in \N^*,
\eeq
from which we get that
$$ \sigma(H)= \sigma_{ac}(H) = \bigcup_{j \in \N^*} \overline{\omega_j(\re)} =  \bigcup_{j \in \N^*} [ (2j-1) b_-,(2j-1) b_+]. $$

\subsection{A positive lower bound on the derivative of the band function}\label{subsec:strict-pos1}

We next obtain a strictly positive lower bound on the derivative $\omega_j'(k)$ for $k$ in selected intervals in $\re$. The edge currents are carried by states in the range of the spectral projector for $H$ and intervals $\Delta_j$ in the energy bands $[(2j-1)b_-, (2j-1)b_+]$, obtained as the range of $\omega_j(k)$, for $k \in \re$. In general, these bands may overlap as seen from \eqref{a9}. In the next lemma, we fix an integer $n$ and show that if $b_\pm$ satisfy a certain relation, then all the bands $\{ \omega_j (k) ~|~ k \in \re \}$, for $j = 1,2, \ldots, n$ do not overlap. This allows us to characterize $\omega_j^{-1} (\Delta_j)$, for certain intervals $\Delta_j \subset \mbox{Ran} ~ \omega_j$.
Finally, we use the fact that $h(k)$ is close to $h_\pm (k)$ for $k \in \R_\pm$, respectively, and that the eigenfunctions and eigenvalues of $h_\pm (k)$ are well known.

\begin{pr}
\label{prop-a2}
Let $b_->0$, $b_+=r b_-$ with $r \in (1,3^{1 \slash 2}]$, and let $n$ denote the unique positive integer satisfying
\bel{a12d}
\left( \frac{2n+3}{2n+1} \right)^{1 \slash 2} < r \leq  \left( \frac{2n+1}{2n-1} \right)^{1 \slash 2}.
\ee Fix $j \in {\mathbb N}_n^*:=\{1,2,\ldots,n \}$
and consider $\Delta_j:=((2j-1+\delta_j) b_-,(2j-1-\delta_j) b_+)$ where $\delta_j$ verifies
\bel{a13}
0 < \delta_j < (2j-1) \left( \frac{r-1}{r+1}\right) < 1/2.
\ee
Then
\begin{itemize}
\item Disjointness of inverse images: $\omega_l^{-1}(\Delta_j)=\emptyset$, for every $l \in {\mathbb N}^* \backslash \{ j \}$.
\item Positivity of the derivative: For each $j \in {\mathbb N}_n^*:=\{1,2,\ldots,n \}$,
there exists a constant $c_j>0$, independent of $k$, $b_{\pm}$ and {$\delta_j$},
such that we have
\bel{a14}
\omega_j'(k) \geq c_j  \delta_j^3 \left( \frac{r-1}{r^3} \right)   b_-^{1 \slash 2},\ k \in \omega_j^{-1}(\Delta_j).
\ee
\end{itemize}
\end{pr}
\begin{proof}
1. Proof of part 1. We first note that $\delta_j$ fulfilling the first and second inequalities in \eqref{a13} satisfies $\delta_j \in (0,1 \slash 2)$. Indeed, we have
$$
\delta_j < (2j-1) (r^2-1) \slash  (r+1)^2 < (2j-1)  ( (2n+1) \slash (2n-1)-1) \slash ( ( (2n+3) \slash (2n+1) )^{1 \slash 2} +1 )^2,
$$
from  \eqref{a12d}-\eqref{a13}, hence
$\delta_j <2 \left( \left( \frac{2n+3}{2n+1} \right)^{1 \slash 2} +1 \right)^{-2}  \left( \frac{2j-1}{2n-1} \right) < 1 \slash 2$.
We next consider $\omega_l^{-1}(\Delta_j)$ as described in the proposition. Due to \eqref{a12d}, we have
$$
b_- \geq \left( \frac{2n-1}{2n+1} \right)^{1 \slash 2} b_+ \geq \left( \frac{2j-1}{2j+1} \right)^{1 \slash 2} b_+ ,
$$
since $j \leq n$. Hence, it follows from this and \eqref{a9} that
$$
\inf_{k \in \re} \omega_{j+1}(k) =(2j+1) b_- > (2j-1) b_+.
$$
This yields $\omega_l^{-1} (\Delta_j)=\emptyset$ for $l \geq j+1$. Similarly, for $n \geq j \geq 2$, it holds true that
$$
\sup_{k \in \re} \omega_{j-1}(k)=(2j-3)b_+ \leq (2j-3) \left( \frac{2n+1}{2n-1} \right)^{1 \slash 2} b_- \leq (2j-3) \left( \frac{2j-1}{2j-3} \right)^{1 \slash 2} b_- <(2j-1) b_-,
$$
so that $\omega_l^{-1}(\Delta_j)=\emptyset$ for $l \in {\mathbb N}_{j-1}^*$.\\

2. Proof of part 2.
To prove the remaining part of this proposition, and the lower bound \eqref{a14}, we examine the two cases $k \leq 0$ and $k \geq 0$ separately.\\

\noindent
{\it Case: $k \leq 0$.} Setting $\alpha_{j,l}^{-}(k):=\langle \psi_j(k) , \psi_l^{-}(k) \rangle$ for all $l \in {\mathbb N}^*$, and using the operator inequality $h(k) \geq h_-(k)$, which holds true for all $k \leq 0$,
we get
$$ 0 \leq \langle \psi_j(k), (h(k)-h_-(k)) \psi_j(k) \rangle = \sum_{l \geq 1} (\omega_j(k) - (2l-1)b_-) | \alpha_{j,l}^{-}(k)|^2,$$
and hence
\beas
\sum_{l=1}^j (\omega_j(k) - (2l-1)b_-) | \alpha_{j,l}^{-}(k)|^2 & \geq & \sum_{l \geq j+1} ((2l-1)b_- -\omega_j(k)) | \alpha_{j,l}^-(k)|^2 \\
& \geq & ((2j+1)b_- - \omega_j(k)) \sum_{l \geq j+1} | \alpha_{j,l}^-(k)|^2.
\eeas
From this and the normalization condition $\sum_{l \geq j+1} | \alpha_{j,l}^-(k)|^2=1-\sum_{l=1}^j | \alpha_{j,l}^-(k)|^2$ then follows that
$2 \sum_{l=1}^j (j+1-l) b_- | \alpha_{j,l}^{-}(k)|^2 \geq (2j+1)b_- -\omega_j(k)$, giving
\bel{a14a}
\sum_{l=1}^j | \alpha_{j,l}^{-}(k)|^2 \geq \frac{(2j+1)b_- - \omega_j(k)}{2j b_-}.
\ee
Bearing in mind that
$\omega_j(k) <(2j-1) \left( \frac{2n+1}{2n-1} \right)^{1 \slash 2} b_- \leq (2j-1) \left( \frac{2j+1}{2j-1} \right)^{1 \slash 2} b_-  < 2j b_-$, we deduce from \eqref{a14a} that
\bel{a14b}
\sum_{l=1}^j | \alpha_{j,l}^{-}(k)|^2 >  \frac{1}{2j},\ k \in \omega_j^{-1}(\Delta_j) \cap \re_-.
\ee
Further, since
$$\langle (h(k)-h_-(k)) \psi_j(k) , \psi_l^-(k) \rangle = \int_0^{+\infty} ((b_+ x-k)^2-(b_- x-k)^2) \psi_j(x,k)  \psi_l^-(x,k) {\rm d} x, $$
for each $l \in {\mathbb N}_j^*$, and since $b_+ x-k \geq b_- x-k \geq 0$ for all $(x,k) \in \re_+ \times \re_-$, we find that
$$
| (\omega_j(k) - (2l-1) b_-) \alpha_{j,l}^-(k)  | \leq  \int_0^{+\infty} (b_+ x-k)^2 |\psi_j(x,k)|  |\psi_l^-(x,k)| {\rm d} x, $$
and hence
\bel{a14c}
| (\omega_j(k) - (2l-1) b_-) \alpha_{j,l}^-(k)  |
\leq \| (b_+ x - k)^{1 \slash 2} \psi_j(k) \|_{{\rm L}^2(\re_+)} \| (b_+ x - k)^{3 \slash 2} \psi_l^-(k) \|_{{\rm L}^2(\re_+)}.
\ee
Now, taking into account that
\bel{a19}
\psi_l^{\pm}(x,k)=\frac{1}{(2^l l!)^{1 \slash 2}}\left( \frac{b_{\pm}}{\pi} \right)^{1 \slash 4} {\rm e}^{-\frac{b_{\pm}}{2} \left( x-\frac{k}{b_{\pm}} \right)^2} H_l \left(b_{\pm}^{1 \slash 2} \left( x - \frac{k}{b_{\pm}} \right) \right),
\ee
where $H_l$ denotes the $l^{\rm th}$ Hermite polynomial, and that $b_+ x-k \leq r (b_- x - k)$ for all $(x,k) \in \re_+ \times \re_-$, we obtain through basic computations that
\bel{a14d}
\| (b_+ x - k)^{3 \slash 2} \psi_l^-(k) \|_{{\rm L}^2(\re_+)} \leq r^{3 \slash 2} \| (b_- x - k)^{3 \slash 2} \psi_l^-(k) \|_{{\rm L}^2(\re_+)}  \leq  \tilde{c}_l  r^{3 \slash 2}
b_-^{3 \slash 4},
\ee
where $\tilde{c}_l:= \frac{1}{(2^l l!)^{1 \slash 2}} \frac{1}{\pi^{1 \slash 4}} \left( \int_0^{+\infty} u^3 {\rm e}^{- u^2} H_l(u)^2 {\rm d} u \right)^{1 \slash 2}>0$ is a constant independent of $b_{\pm}$ and $k$.
Further, as $\| (b_+ x - k)^{1 \slash 2} \psi_j(k) \|_{{\rm L}^2(\re_+)} = \left( \frac{\omega_j'(k)}{2 ( r-1 )} \right)^{1 \slash 2}$ by \eqref{a12c}, we deduce from \eqref{a14c} and
\eqref{a14d} that
$$
\omega_j'(k) \geq 2 \tilde{c}_l^{-2} |\alpha_{j,l}^-(k) |^2 \left( \frac{r-1}{r^3} \right)  (\omega_j(k) - (2l-1) b_-)^2  b_-^{-3 \slash 2} ,\ l=1,2,\ldots,j.
$$
Now, by summing up the above estimate over $l=1,2,\ldots,j$, minorizing $\omega_j(k) -(2l-1)b_-$ by $\delta_j b_-$ for every $l$, and recalling \eqref{a14b}, we end up getting that
\bel{a14e}
\omega_j'(k) \geq c_j \delta_j^2 \left( \frac{r-1}{r^3} \right)  b_-^{1 \slash 2} ,\ k \in \omega_j^{-1}(\Delta_j) \cap \re_-,
\ee
with $c_j:= (\max_{1 \leq l \leq j} \tilde{c}_l)^{-2} \slash j>0$.

\noindent
{\it Case $k \geq 0$.}
Notice that $r^2 V(x,k) = \left( r k - b_+ x \right)^2 \geq V_+(x,k)$
for all $(x,k) \in \re_- \times \re_+$, so we have $r^2 h(k) \geq  h_+(k)$ in the operatorial sense, and consequently
$$ \sum_{l \geq 1} ( r^2 \omega_j(k) - (2l-1) b_+ ) | \alpha_{j,l}^+(k) |^2  \\
= \langle ( r^2 h(k) - h_+(k) ) \psi_j(k) , \psi_j(k) \rangle \geq 0, $$
where $\alpha_{j,l}^+(k):= \langle \psi_j(k), \psi_l^+(k) \rangle$. This yields
\bel{a15}
\sum_{l \geq j+1} \left( (2l-1) b_+ - r^2 \omega_j(k) \right) | \alpha_{j,l}^+(k) |^2
\leq \sum_{l=1}^j ( r^2 \omega_j(k) - (2l-1) b_+ ) | \alpha_{j,l}^+(k) |^2.
\ee
Further, as
$r^2 \omega_j(k) < \frac{2n+1}{2n-1} (2j-1-\delta_j) b_+ \leq \frac{2j+1}{2j-1} (2j-1-\delta_j) b_+$ for every $k \in \omega_j^{-1}(\re)$, we obtain simultaneously that
$r^2 \omega_j(k) - (2l-1) b_+ \leq 2j b_+$ for every $l \in {\mathbb N}_j^*$, and that
$(2l-1) b_+ - r^2 \omega_j(k) \geq \frac{2j+1}{2j-1} \delta_j b_+$ for all $l \geq j+1$.
From this and \eqref{a15} then follows that
$$
\frac{2j+1}{2j-1} \delta_j \sum_{l= j+1}^{\infty} | \alpha_{j,l}^+(k) |^2 \leq 2 j \sum_{l=1}^j | \alpha_{j,l}^+(k) |^2,
$$
which, together with the normalization condition $\sum_{l= j+1}^{\infty} | \alpha_{j,l}^+(k) |^2=1- \sum_{l=1}^j | \alpha_{j,l}^+(k) |^2$, and the inequality $0<\delta_j<2j-1$, arising from \eqref{a13}, entails
\bel{a16}
\sum_{l=1}^j | \alpha_{j,l}^+(k) |^2 \geq \frac{\delta_j}{2(2j-1)},\ k \in \omega_j^{-1}(\Delta_j) \cap \re_+.
\ee
Now, using that
\beas
((2l-1) b_+ - \omega_j(k)) \alpha_{j,l}^+(k)
& = & \langle (h_+(k) - h(k)) \psi_j(k), \psi_l^+(k) \rangle  \\
& = &  \int_{-\infty}^0 \left( V_+(x,k)- V(x,k) \right) \psi_j(x,k) \psi_l^+(x,k) {\rm d} x,
\eeas
for all $l \in {\mathbb N}^*$, and bearing in mind that $V_+(x,k) \geq V_-(x,k) \geq 0$ for every $(x,k) \in \re_- \times \re_+$,
we get that
$$|(2l-1) b_+ - \omega_j(k) | | \alpha_{j,l}^+(k) |
\leq  \int_{-\infty}^0 V_+(x,k) | \psi_j(x,k) | | \psi_l^+(x,k) | {\rm d} x.
$$
This, together with the elementary estimate $0 \leq k-b_+ x \leq  r(k-b_-x)$, which holds true for every $(x,k) \in \re_- \times \re_+$, shows that each $| (2l-1) b_+ - \omega_j(k) | | \alpha_{j,l}^+(k) |$, $l \in {\mathbb N}_j^*$, is majorized by the scalar product
$r^{1 \slash 2} \int_{-\infty}^0 (k-b_-x)^{1 \slash 2} | \psi_j(x,k) | (k-b_+ x)^{3 \slash 2} | \psi_l(x,k) | {\rm d} x$. Therefore, we have
\bea
& & | (2l-1) b_+ - \omega_j(k) | | \alpha_{j,l}^+(k) | \nonumber \\
& \leq & r^{1 \slash 2} \| (k-b_-x)^{1 \slash 2} \psi_j(k) \|_{{\rm L}^2(\re_-)}
\| (k-b_+ x)^{3 \slash 2} \psi_l^+(k) \|_{{\rm L}^2(\re_-)},\ l \in {\mathbb N}_j^*, \label{a18}
\eea
by applying the Cauchy-Schwarz inequality.
Next, using standard computations, we derive from the explicit expression \eqref{a19} of $\psi_l^+(k)$, that
\bel{a20}
\| (k-b_+ x)^{3 \slash 2} \psi_l^+(k) \|_{{\rm L}^2(\re_-)} = \tilde{c}_l b_+^{3 \slash 4},\ l \in {\mathbb N}_j^*,
\ee
where $\tilde{c}_l := \frac{1}{(2^l l!)^{1 \slash 2}} \frac{1}{\pi^{1 \slash 4}} \left( \int_0^{+\infty} u^3 {\rm e}^{-u^2} H_l^2(-u) {\rm d}u \right)^{1 \slash 2}>0$ is a constant depending only on $l$. Furthermore, we have
$\omega_j'(k) =2  \left( \frac{r-1}{r} \right) \| (k-b_-x)^{1 \slash 2} \psi_j(k) \|_{{\rm L}^2(\re_-)}^2$ by \eqref{a12b}, so we obtain
\bel{a21}
\omega_j'(k) \geq  2 \tilde{c}_l ^{-2}  | \alpha_{j,l}^+(k)|^2 \left( \frac{r-1}{r^2} \right) ( (2l-1) b_+ - \omega_j(k)  )^2  b_+^{-3\slash 2},\ l \in {\mathbb N}_j^*,
\ee
directly from \eqref{a18}-\eqref{a20}.
Actually, it holds true that
\bel{a21b}
| (2l-1) b_+ - \omega_j(k) | > \delta_j b_+,\ l \in {\mathbb N}_j^*,\ k \in \omega_j^{-1}(\Delta_j) \cap \re_+.
\ee
Indeed, \eqref{a21b} follows immediately from the inequality $\omega_j(k) < (2j-1-\delta_j) b_+$ for $l=j$, and from the two estimates $\omega_j(k) > (2j-1+\delta_j) b_-$ and
$$(2l-1)b_+ \leq (2j-3) b_+ \leq (2j-3) \left( \frac{2n+1}{2n-1} \right)^{1 \slash 2} b_- \leq (2j-3) \left( \frac{2j-1}{2j-3} \right)^{1 \slash 2} b_- < (2j-1) b_-,$$
when $l \in {\mathbb N}_{j-1}^*$ and $n \geq j \geq 2$.
From \eqref{a21}-\eqref{a21b} then follows that
\bel{a22}
\omega_j'(k) \geq 2 \tilde{c}_l^{-2} | \alpha_{j,l}^+(k) |^2  \delta_j^2 \left( \frac{r-1}{r^{3 \slash 2}} \right)  b_-^{1 \slash 2} \ l \in {\mathbb N}_j^*,\ k \in \omega_j^{-1}(\Delta_j) \cap \re_+.
\ee
Finally, by summing up \eqref{a22} over $l=1,2,\ldots,j$, and putting the result together with \eqref{a16}, we end up getting that
$$
\omega_j'(k) \geq c_j \delta_j^3 \left( \frac{r-1}{r^{3 \slash 2}} \right) b_-^{1 \slash 2},\ k \in \omega_j^{-1}(\Delta_j) \cap \re_+,
$$
where $c_j:=(\max_{1 \leq l \leq j} \tilde{c}_l)^{-2} \slash (2j-1)>0$.
Now \eqref{a14} follows immediately from this and from \eqref{a14e}.
\end{proof}


\section{Existence and localization of edge currents}\label{sec:edge-currents1}
\setcounter{equation}{0}

In light of \cite{CHS,HS1,HS2}, we define the current carried by a state $\varphi$ as the expectation of the $y$-component of the velocity operator $v_y:= (i/2)[H,y] = p_y-\beta(x)$ in the state $\varphi$, i.e. $J_y(\varphi):= \langle v_y \varphi , \varphi \rangle $.
In this section, we prove the existence of states carrying an edge current and its localization near $x=0$.

\subsection{Edge states carrying a current}

We expect that a state with energy localized in intervals away from the Landau levels for $b_\pm$ will carry a current. We prove this by establishing a lower bound on the matrix element $\langle v_y \varphi , \varphi \rangle$.

\begin{theorem} 
\label{lm-a3}
Let $b_-$, $r$, $n$, $j$, $\delta_j$ and $\Delta_j$ be as in Proposition \ref{prop-a2}, and let $\varphi \in {\rm L}^2(\re^2)$ satisfy $\varphi=\mathbb{P}(\Delta_j) \varphi$, where $\mathbb{P}(I)$ denotes the spectral projection of $H$ for the Borel set $I \subset \re$. We have the following estimate for $j = 1 , \ldots, n$,
\beq\label{eq:lower-bound1}
J_y(\varphi)  \geq  c_j \delta_j^3 \left( \frac{r-1}{r^3} \right)  b_-^{1 \slash 2} \| \varphi \|^2
 = c_j \delta_j^3 \left( \frac{b_+ - b_-}{b_-} \right) \left( \frac{b_-}{b_+} \right)^3 b_-^{1 \slash 2} \| \varphi \|^2 ,
\eeq
where $c_j$ is the constant introduced in \eqref{a14}.
\end{theorem}

\begin{proof}
The proof depends on the identity
$$
J_y(\varphi)=\int_{\re^2} \hat{v}_y(k) | \hat{\varphi}(x,k) |^2 {\rm d} x\ {\rm d} k  ,
$$
with $\hat{v}_y(k):=k- \beta(x)$. Since the state $\varphi$ satisfies $\varphi = \mathbb{P}(\Delta_j) \varphi$,
its partial Fourier transform may be written as
\bel{a22b}
\hat{\varphi}(x,k):=(\mathcal{F} \varphi)(x,k) = \chi_{\omega_j^{-1}(\Delta_j)}(k) \beta_j(k) \psi_j(x,k),
\ee
where $\chi_I$ denotes the characteristic function of $I \subset \re$ and $\beta_j(k):=\langle \hat{\varphi}(k) , \psi_j(k) \rangle_{{\rm L}^2(\re)}$. This yields that
\beq\label{eq:current99}
J_y(\varphi) = (1 \slash 2) \int_{\omega_j^{-1}(\Delta_j)} \omega_j'(k) | \beta_j(k) |^2 {\rm d} k ,
\eeq
so the result follows immediately from Proposition \ref{prop-a2}.
\end{proof}

\begin{remark}
In the case $b_- = b_+$, we have $\Delta_j = \emptyset$. By \eqref{eq:current99} this implies that $J_y(\varphi) = 0$ so there is no edge current. This is consistent with the fact that
the Landau Hamiltonian has only pure point spectrum with localized eigenfunction.
\end{remark}

\begin{remark}
The dependance of the lower bound on $b_-$ is optimal in the following sense. Recall from Proposition \ref{prop-a2} that $\Delta_j
= ((2j-1+\delta_j)b_-, (2j-1- \delta_j)b_+)$. For $\varphi=\mathbb{P}(\Delta_j) \varphi$,
we have $\langle H \varphi, \varphi \rangle \in \Delta_j$ implying that $\| (p_y - \beta (x)) \varphi \| \leq C_1 b_+^{1/2}$,
for a constant $C_1 > 0$ independent of $b_\pm$. Hence, there exists a constant $C_2 > 0$, independent of $b_\pm$,
so that we have the upper bound on the current
\beq\label{eq:upper-bound1}
| \langle \varphi, v_y \varphi \rangle | \leq C_2 b_+^{1/2}.
\eeq
\end{remark}

\begin{remark}
The parameters $b_\pm$, $n$ and $j$ being the same as in Theorem \ref{lm-a3}, an estimate of the kind of \eqref{eq:lower-bound1}, i.e.
$$ \exists C_j>0,\ J(\varphi) \geq C_j b_-^{1 \slash 2} \| \varphi \|^2,\ \varphi ={\mathbb P}(\Delta_j) \varphi, $$
is no longer valid when the open interval $\Delta_j$ contains either $(2j-1)b_-$ or $(2j-1)b_+$. This can be seen from Lemma \ref{lm-a1}, \eqref{eq:bflimit1} and the analyticity of
$k \mapsto \omega_j(k)$, entailing the existence of $k_j \in \R$ and $\kappa>0$ such that
$$\omega_j(k) \in \Delta_j\ \mbox{and}\ \omega_j'(k) \in (0,C_j b_-^{1 \slash 2}),\ k \in I_j:=(k_j-\kappa,k_j+\kappa), $$
so the state $\varphi_0(x,y)=(2 \kappa)^{-1 \slash 2} \int_{\R \times I_j} e^{iky} \psi_j(x,k) dx dk \in \mathbb{P}(\Delta_j) L^2(\R^2)$ verifies
$$ J(\varphi_0) = \frac{1}{2 \kappa}\int_{I_j} \omega_j'(k) dk < C_j b_-^{1 \slash 2} \| \varphi_0 \|^2. $$
\end{remark}

\subsection{Localization of the edge currents}

Edge currents correspond to the trajectories of a classical charged particle moving under the influence of the magnetic field $b(x)$. The classical cyclotron radius is $b^{-1/2}$ so a particle starting at $x=0$ and with a positive velocity will move in the $x>0$ half plane in a circular orbit with radius $b_+^{-1/2}$. When it reaches $x=0$, the radius of the orbit changes to $b_-^{-1/2}$. Since $b_+ > b_-$, there is a net flow in the negative $y$ direction in a spiral orbit. The classical particle is constrained to a strip of width $\sim 2 b_-^{-1/2}$ about $x=0$. We prove that the quantum edge currents described in Theorem \ref{lm-a3} are likewise constrained to a small strip about $x=0$.

\begin{theorem} 
\label{lm-a4}
Let $b_-$, $b_+$, $n$, $j$ and $\Delta_j$ be as in Proposition \ref{prop-a2}. Then for all $\varepsilon_1>0$ and $\varepsilon_2>0$ there
exists $b_j(\varepsilon_1,\varepsilon_2)>0$ such that any ${\rm L}^2(\re^2)$-normalized function $\varphi = {\mathbb P}(\Delta_j) \varphi$ satisfies
$$ \int_{\re^2} \chi_{I_{\varepsilon_1,\varepsilon_2}}(x) | \varphi(x,y) |^2 {\rm d} x {\rm d} y
\geq  1 - \eta_j {\rm e}^{-\varepsilon_1^2 b_-^{2 \varepsilon_2} \slash 8}, $$
provided $b_- \geq b_j(\varepsilon_1,\varepsilon_2)$. Here $\chi_{I_{\varepsilon_1,\varepsilon_2}}$ denotes the characteristic function of the interval
$I_{\varepsilon_1,\varepsilon_2}:= [-(1+\varepsilon_1) b_-^{-1 \slash 2+\varepsilon_2},(1+\varepsilon_1) b_+^{-1 \slash 2 + \varepsilon_2}]$ and $\eta_j:=2 (\pi (2j-1))^{1 \slash 2}$.
\end{theorem}
\begin{proof}
The proof consists of three steps.\\
{\it First step.} We first show that
\bel{a25}
\forall \varepsilon_2 >0, \exists b(\varepsilon_2)>0,\ b_- \geq b(\varepsilon_2) \Longrightarrow  \omega_j^{-1}(\Delta_j) \subset (-b_-^{1 \slash 2 + \varepsilon_2}, b_+^{1 \slash 2 + \varepsilon_2}).
\ee
We shall actually only prove that $\omega_j^{-1}(\Delta_j) \subset (-b_-^{1 \slash 2 + \varepsilon_2},+\infty)$, the remaining part of the proof being obtained in a similar way.
To do that we introduce $\chi \in {\rm C}^2(\re;[0,1])$ satisfying
$$ \chi(x) = \left\{ \begin{array}{cl} 0 & {\rm for}\ x \leq -b_-^{-1 \slash 2 + 2 \varepsilon_2} \\ 1 & {\rm for}\ x \geq -b_-^{-1 \slash 2 + 2 \varepsilon_2} \slash 2, \end{array} \right. $$
and $\| \chi' \| \leq c b_-^{1 \slash 2}$, $\| \chi'' \| \leq c b_-$, for some constant $c>0$ independent of $b_{\pm}$. Then we set $k=k_{\varepsilon_2}:=-b_-^{1 \slash 2 + \varepsilon_2}$ and deduce from \eqref{a19} that
\beas
\| \chi \psi_j^-(k_{\varepsilon_2}) \|^2 & \geq & C_j b_-^{1 \slash 2} \int_{-b_-^{-1 \slash 2 + 2 \varepsilon_2} \slash 2}^{+\infty} {\rm e}^{-b_-(x-k_{\varepsilon_2} \slash b_-)^2} H_j(b_-^{1 \slash 2}(x- k_{\varepsilon_2} \slash b_-))^2 {\rm d} x  \\
& \geq & C_j \int_{-b_-^{2 \varepsilon_2} \slash 2 + b_-^{\varepsilon_2}}^{+\infty} {\rm e}^{-u^2} H_j(u)^2 {\rm d} u,
\eeas
for some constant $C_j>0$ depending only on $j$. Taking $b_-$ so large that  $-b_-^{2 \varepsilon_2} \slash 2 + b_-^{\varepsilon_2} \leq 0$, we thus find that
\bel{a26}
\| \chi \psi_j(k_{\varepsilon_2}) \|^2  \geq C_j \int_{0}^{+\infty} {\rm e}^{-u^2} H_j(u)^2 {\rm d} u >0.
\ee
Further we notice that
\beas
(h(k_{\varepsilon_2})-(2j-1) b_-) \chi \psi_j^-(k_{\varepsilon_2}) & = & [ h(k_{\varepsilon_2}), \chi ] \psi_j^-(k_{\varepsilon_2}) - \chi (h(k_{\varepsilon_2})-h_-(k_{\varepsilon_2}))\psi_j^-(k_{\varepsilon_2}) \\
& = & -2 {\rm i} \chi' (\psi_j^-(k_{\varepsilon_2}))' - \chi'' \psi_j^-(k_{\varepsilon_2}) - \chi (V(k_{\varepsilon_2})-V_-(k_{\varepsilon_2})) \psi_j^-(k_{\varepsilon_2}).
\eeas
As $\chi(x) (V(x,k_{\varepsilon_2})-V_-(x,k_{\varepsilon_2})) = (V_+(x,k_{\varepsilon_2})-V_-(x,k_{\varepsilon_2})) \chi_{\re_+}(x)$, this implies that
\bel{a27}
(h(k_{\varepsilon_2})-(2j-1) b_-) \chi \psi_j^-(k_{\varepsilon_2}) = -2 {\rm i} \chi' (\psi_j^-(k_{\varepsilon_2}))' - \chi'' \psi_j^-(k_{\varepsilon_2}) - (V_+(k_{\varepsilon_2})-V_-(k_{\varepsilon_2})) \chi_{\re_+} \psi_j^-(k_{\varepsilon_2}).
\ee
Actually, due to \eqref{a19} we find that
\beas
\| \chi'' \psi_j^-(k_{\varepsilon_2}) \|^2 & \leq & C_j'' b_-^{5 \slash 2}  \int_{-b_-^{-1 \slash 2 + 2 \varepsilon_2}}^{-b_-^{-1 \slash 2 + 2 \varepsilon_2} \slash 2} {\rm e}^{-b_-(x-k_{\varepsilon_2} \slash b_-)^2} H_j(b_-^{1 \slash 2}(x- k_{\varepsilon_2} \slash b_-))^2 {\rm d} x \\
& \leq & C_j'' b_-^2 \int_{-b_-^{2 \varepsilon_2} + b_-^{\varepsilon_2}}^{-b_-^{2 \varepsilon_2} \slash 2 + b_-^{\varepsilon_2}} {\rm e}^{-u^2} H_j(u)^2 {\rm d} u \\
& \leq & C_j'' b_-^2 \int_{-\infty}^{-b_-^{\varepsilon_2}} {\rm e}^{-u^2} H_j(u)^2 {\rm d} u,
\eeas
provided $b_-$ is taken so large that $b_-^{\varepsilon_2} \geq 4$, the constant $C_j''>0$ depending only on $j$.
Hence
\bel{a28}
\| \chi'' \psi_j^-(k_{\varepsilon_2}) \|^2 \leq C_j'' b_-^2 {\rm e}^{-b_-^{2 \varepsilon_2} \slash 2} \int_{-\infty}^{-b_-^{ \varepsilon_2}} {\rm e}^{-u^2 \slash 2} H_j(u)^2 {\rm d} u.
\ee
By reasoning in the same way with $\chi' (\psi_j^-(k_{\varepsilon_2}))'$, we obtain that
\bel{a29}
\| \chi' (\psi_j^-(k_{\varepsilon_2}))' \|^2 \leq C_j' b_- {\rm e}^{-b_-^{2 \varepsilon_2} \slash 2} \int_{-\infty}^{-b_-^{ \varepsilon_2}} {\rm e}^{-u^2 \slash 2} (H_j'(u)-uH_j(u))^2 {\rm d} u,
\ee
where the constant $C_j'>0$ depends only on $j$. Finally, as
$$ \| (V_+(k_{\varepsilon_2})-V_-(k_{\varepsilon_2})) \psi_j^-(k_{\varepsilon_2}) \|_{{\rm L}^2 (\re_+)}^2 \leq  \int_0^{+\infty}
V_+(x,k_{\varepsilon_2})^2 \psi_j^-(x,k_{\varepsilon_2})^2 {\rm d} x, $$
and $V_+(x,k_{\varepsilon_2}) \leq b_+^2 (x-k_{\varepsilon_2} \slash b_-)^2$ for all $x \geq 0$, we deduce from \eqref{a19}
that
\beas
& & \| (V_+(k_{\varepsilon_2})-V_-(k_{\varepsilon_2})) \psi_j^-(k_{\varepsilon_2}) \|_{{\rm L}^2 (\re_+)}^2 \\
& \leq & c_j b_+^4  \int_0^{+\infty} b_-^{1 \slash 2} (x - k_{\varepsilon_2} \slash b_-)^4 {\rm e}^{-b_-(x-k_{\varepsilon_2} \slash b_-)^2} H_j(b_-^{1 \slash 2} (x-k_{\varepsilon_2} \slash b_-))^2 {\rm d} x \\
& \leq & c_j (b_+^2 \slash b_-)^2 \int_{b_-^{\varepsilon_2}}^{+\infty}  u^4 {\rm e}^{-u^2} H_j(u)^2 {\rm d} u \\
& \leq & c_j (b_+^2 \slash b_-)^2 {\rm e}^{-b_-^{2 \varepsilon_2} \slash 2} \int_{0}^{+\infty}  u^4 {\rm e}^{-u^2 \slash 2} H_j(u)^2 {\rm d} u,
\eeas
the constant $c_j$ depending only on $j$.
From this and \eqref{a27}-\eqref{a29} then follows that
$$ \| h(k_{\varepsilon_2})-(2j-1) b_- \| \leq C_{n,j} b_- {\rm e}^{-b_-^{2 \varepsilon_2} \slash 4}, $$
provided $b_-$ is large enough, where $C_{n,j}>0$ depends only on $n$ and $j$. In light of \eqref{a26} this entails that
${\rm dist}(\sigma(h(k_{\varepsilon_2})),(2j-1) b_-)$ can be made smaller than $\delta_j b_-$ upon choosing $b_-$ sufficiently large. Bearing in mind that
$$
\omega_l(k_{\varepsilon_2}) \leq \omega_{j-1}(k_{\varepsilon_2}) < (2j-3)b_+  \leq  (2j-3) \left( \frac{2n+1}{2n-1} \right)^{1 \slash 2} b_-,$$
for all $l \leq j-1$ (when $n \geq j \geq 2$) so that
$$ (2j-1) b_- - \omega_l(k_{\varepsilon_2}) > (2j-1) b_- - (2j-3) \left( \frac{2j-1}{2j-3} \right)^{1 \slash 2} b_- > b_- > 2 \delta_j b_-,\ l \leq j-1, $$
and that
$$ \omega_l(k_{\varepsilon_2}) - (2j-1) b_- \geq \omega_{j+1}(k_{\varepsilon_2}) - (2j-1) b_- > 2 b_- > 4 \delta_j b_-,\ l \geq j+1, $$
we necessarily have $0 < \omega_j(k_{\varepsilon_2}) -(2j-1)b_- < \delta_j b_-$. This yields $\omega_j(k_{\varepsilon_2}) < \inf \Delta_j$, and hence $\omega_j(k) < \inf \Delta_j$ for all $k \leq k_{\varepsilon_2}$ according to Lemma \ref{lm-a1}, so that
$\omega_j^{-1}(\Delta_j) \subset (k_{\varepsilon_2},+\infty)$. \\

\noindent
{\it Second step.} Choose $b_- \geq b(\varepsilon_2)$ so that \eqref{a25} holds true. We will prove that an eigenfunction $\psi_j$
decays in the regions $\pm x \geq  \pm x_j^\pm (\epsilon_2)$. In particular, we will prove
\bel{a30}
|\psi_j(x,k)| \leq 2^{1 \slash 2} (2j-1)^{1 \slash 4} b_+^{1 \slash 4} {\rm e}^{-b_{\pm}(x-x_j^{\pm}(\varepsilon_2))^2 \slash 2},\ \pm x \geq \pm x_j^{\pm}(\varepsilon_2),\ k \in \omega_j^{-1}(\Delta_j),
\ee
where
\bel{a31}
x_j^{\pm}(\varepsilon_2):= \pm \left( b_{\pm}^{-1 \slash 2 + \varepsilon_2}+ (2j-1)^{1 \slash 2} \frac{b_+^{1 \slash 2}}{b_{\pm}} \right).
\ee
To prove \eqref{a30}-\eqref{a31}, we use \eqref{a25} and check that for all $k \in \omega_j^{-1}(\Delta_j)$
and that for every $\pm x > \pm x_j^{\pm}(\varepsilon_2)$,
$$
Q_j(x,k):=V(x,k)-\omega_j(k) \geq b_{\pm}^2(x-x_j^{\pm}(\varepsilon_2))^2>0 .
$$
Using this positivity and integrating the differential equation $\psi_j'' = Q_j \psi_j$ over the regions $\pm x \geq \pm x_j^\pm (\epsilon_2)$,
we establish that $\psi_j' \psi_j$ has a fixed sign in each region. This implies that $\psi_j' / \psi_j = ( \psi_j' \psi_j) / \psi_j^2$ has the same sign in the same regions. Following Iwatsuka \cite[Lemma 3.5]{iwatsuka1}, since $\psi_j'' = Q_j \psi_j$, differentiating $(\psi_j')^2 - Q_j \psi_j^2$, one finds that it is negative since $Q_j' > 0$ in the regions. Since $(\psi_j')^2 - Q_j \psi_j^2$ vanishes at infinity, this means that it is positive from which we conclude that $(\psi_j')^2 \geq Q_j \psi_j^2$. Summarizing these arguments, we obtain
$$\psi_j'(x,k) \psi_j(x,k) < 0\ {\rm and}\ \frac{\psi_j'(x,k)}{\psi_j(x,k)} < - Q_j(x,k)^{1 \slash 2},\ x \geq x_j^+(\varepsilon_2), $$
and
$$\psi_j'(x,k) \psi_j(x,k) > 0\ {\rm and}\ \frac{\psi_j'(x,k)}{\psi_j(x,k)} > Q_j(x,k)^{1 \slash 2},\ x \leq x_j^-(\varepsilon_2). $$
Integrating the inequalities involving $Q_j$ over each region we obtain
$$| \psi_j(x,k) | \leq | \psi_j(x_j^+(\varepsilon_2),k) | {\rm e}^{-\int_{x_j^+(\varepsilon_2)}^x b_+ (t - x_j^{+}(\varepsilon_2)) {\rm d} t},\ x \geq x_j^+(\varepsilon_2),$$
and
$$| \psi_j(x,k) | \leq | \psi_j(x_j^-(\varepsilon_2),k) | {\rm e}^{-\int_x^{x_j^{-}(\varepsilon_2)} b_- (x_j^{-}(\varepsilon_2)-t) {\rm d} t},\ x \leq x_j^-(\varepsilon_2).$$
The result \eqref{a30} follows from this and the following estimate
$$ \psi_j(x,k)^2 \leq 2 \left( \int_{-\infty}^x \psi_j'(t,k)^2 {\rm d} t \right)^{1 \slash 2} \leq 2 \omega_j(k) < 2(2j-1) b_+,\ x \in \re.$$
\\

\noindent
{\it Third step.}
Choose $b_-$ so large that \eqref{a25} and
$\varepsilon_1 b_+^{\varepsilon_2} \geq 2(2j-1)^{1 \slash 2}$ hold simultaneously true. Notice that this last condition actually guarantees that we have
$(1+\varepsilon_1 \slash 2) b_+^{-1 \slash 2 +\varepsilon_2} \geq  x_j^+(\varepsilon_2) + (\varepsilon_1 \slash 2)) b_+^{-1 \slash 2 +\varepsilon_2}$.
This and \eqref{a30} yields for every $k \in \omega_j^{-1}(\Delta_j)$ that
\beas
\int_{(1+\varepsilon_1) b_+^{-1 \slash 2 +\varepsilon_2}}^{+\infty} \psi_j(x,k)^2 {\rm d} x & \leq &
2(2j-1)^{1 \slash 2} b_+^{1 \slash 2} \int_{x_j^+(\varepsilon_2) + (\varepsilon_1 \slash 2) b_+^{-1 \slash 2 + \varepsilon_2}}^{+\infty}
{\rm e}^{-b_+(x-x_j^+(\varepsilon_2))^2} {\rm d} x \\
& \leq & 2(2j-1)^{1 \slash 2} \int_{(\varepsilon_1 \slash 2) b_+^{\varepsilon_2}}^{+\infty} {\rm e}^{-u^2} {\rm d} u,
\eeas
hence
$$
\int_{(1+\varepsilon_1) b_+^{-1 \slash 2 +\varepsilon_2}}^{+\infty} \psi_j(x,k)^2 {\rm d} x
\leq (\pi(2j-1))^{1 \slash 2} {\rm e}^{-\varepsilon_1^2 b_+^{2 \varepsilon_2} \slash 8},\ k \in \omega_j^{-1}(\Delta_j).
$$
By reasoning in the same way we find out for every $k \in \omega_j^{-1}(\Delta_j)$ (upon choosing $b_-$ so large that $\varepsilon_1 b_-^{\varepsilon_2} \geq 2 (2j-1)^{1 \slash 2} (b_+ \slash b_-)^{1 \slash 2}$, that is $(1+\varepsilon_1) b_-^{-1 \slash 2 +\varepsilon_2} \geq  -x_j^-(\varepsilon_2) + (\varepsilon_1 \slash 2) b_-^{-1 \slash 2 + \varepsilon_2}$) that the integral
$\int_{-\infty}^{-(1+\varepsilon_1) b_-^{-1 \slash 2 +\varepsilon_2}} \psi_j(x,k)^2 {\rm d} x$ is majorized by
$(\pi(2j-1))^{1 \slash 2} {\rm e}^{-\varepsilon_1^2 b_-^{2 \varepsilon_2} \slash 8}$ , so we get
\bel{a32}
\int_{\re \backslash I_{\varepsilon_1,\varepsilon_2}} \psi_j(x,k)^2 {\rm d} x \leq 2 (\pi(2j-1))^{1 \slash 2} {\rm e}^{-\varepsilon_1^2 b_-^{2 \varepsilon_2} \slash 8},\ k \in \omega_j^{-1}(\Delta_j).
\ee
Finally, by recalling \eqref{a22b}, we have
\beas
\int_{\re^2} \chi_{I_{\varepsilon_1,\varepsilon_2}}(x) | \varphi(x,y) |^2 {\rm d} x {\rm d} y &
= & \int_{\re^2}  \chi_{I_{\varepsilon_1,\varepsilon_2}}(x) | \hat{\varphi}(x,k) |^2 {\rm d} x {\rm d} k \\
& = & \int_{\omega_j^{-1}(\Delta_j)} | \beta_j(k) |^2 \left( \int_{I_{\varepsilon_1,\varepsilon_2}} \psi_j(x,k)^2 {\rm d} x \right) {\rm d} k,
\eeas
which, combined with \eqref{a32} and the identity $\int_{\omega_j^{-1}(\Delta_j)} | \beta_j(k)|^2 {\rm d} k =1$,
yields the desired result.
\end{proof}


\section{Smooth Iwatsuka Hamiltonians with positive magnetic fields}\label{sec:smooth1}
\setcounter{equation}{0}

Having completed the analysis of Iwatsuka Hamiltonians with discontinuous magnetic fields, we turn to the case when the magnetic field
is everywhere bounded $0 < b_- \leq b(x) \leq b_+ < \infty$ and assumes constant values outside of an interval $[- \eps, \eps]$.
We have $b(x) = b_-$ for $x < - \epsilon$ and $b(x) = b_+$ for $x > \epsilon$.
We will rely on the results obtained in the previous sections for $\eps = 0$, and show how they imply analogous results in this case. We will take $\epsilon > 0$ small, on the order of $b_-^{-1/2}$. This will insure that the edge currents remain well-localized in a strip around $x=0$.

The magnetic field is defined as follows. Given $\eps \geq 0$ and $0<b_-<b_+<\infty$, we consider $b_{\eps} \in {\rm L}^1_{\rm loc}(\re)$
satisfying
$$
\left\{ \begin{array}{ll} b_{\eps}(x):=b_{\pm}, & {\rm if}\ \pm x > \eps, \\ b_- \leq b(x) \leq b_+, & {\rm when}\ |x| \leq \eps. \end{array} \right. $$
As in \eqref{a1}, we set
\bel{b1}
\beta_{\eps}(x):=\int_0^x b_{\eps}(s) ds,\ x \in \re
\ee
and consider the 2D vector potential $A_{\eps}:=(A_{\eps,1},A_{\eps,2})$ defined by $A_{\eps,1}:=0$ and $A_{\eps,2}:=\beta_{\eps}(x)$.

The 2D magnetic Schr\"odinger operator $H(A_{\eps})$
is defined on the dense domain ${\rm C}_0^\infty (\re^2)$ by
\bel{b2}
H_{\eps}=H(A_{\eps}) := (-{\rm i} \nabla - A_{\eps})^2 = p_x^2 + (p_y - \beta_{\eps}(x))^2.
\ee
As in section \ref{sec:sharp1}, the partial Fourier transform leads to a direct integral composition
\bel{b4}
{\mathcal F} H_{\eps} {\mathcal F}^*= \int^{\oplus}_{\re} h_{\eps}(k) dk
\ee
with
\bel{b5}
h_{\eps}(k) := p_x^2 + V_{\eps}(x,k)\ {\rm on}\ {\rm L}^2 (\re),\
\ {\rm and} \ V_{\eps}(x,k) := (k-\beta_{\eps}(x))^2.
\ee

In light of \eqref{b1}, the potential $V_{\eps}(x,k)$, $k \in \re$, is unbounded as $|x|$ goes to infinity, hence $h_{\eps}(k)$ has a compact resolvent.
Let $\left\{\omega_{\eps,j}(k)\right\}_{j=1}^{\infty}$ be the increasing
sequence of the eigenvalues of the operator $h_{\eps}(k)$, $k \in \re$.
Since all the eigenvalues $\omega_{\eps,j}(k)$ are simple they depend
analytically on $k \in \re$. Moreover, for all $k \in \re$ there is a unique $x_{\eps,k} \in \re$ such that $\beta_{\eps}(x_{\eps,k})=k$ since $\beta_{\eps}'(x)=b_{\eps}(x) \geq b_-$ for every $x \in \re$. As a consequence we have
$$ b_-^2(x-x_{\eps,k})^2 \leq V_{\eps}(x,k) \leq b_+^2 (x-x_{\eps,k})^2,\ x \in \re, $$
whence
\bel{b9}
(2j-1) b_- \leq \omega_{\eps,j}(k) \leq (2j-1)b_+,\ j \in {\mathbb N}^*,
\ee
from the minimax principle.

Further, let $\left\{\psi_{\eps,j}(k)\right\}_{j=1}^{\infty}$ be the ${\rm L}^2(\re)$-normalized eigenfunctions of
$h_\eps (k)$ satisfying
\bel{b10b}
(h_{\eps}(k) \psi_{\eps,j})(x,k) = \omega_{\eps,j}(k) \psi_{\eps,j}(x,k),\ x \in \re.
\ee
We choose all $\psi_{\eps,j}(k)$ to be real, and $\psi_{\eps,1}(x,k) > 0$, for $x \in
\re$ and $k \in \re$. Since $V_{\eps}(.,k) \in {\rm C}^0(\re) \cap {\rm C}^\infty (\R^*)$,
the functions $\psi_{\eps,j}(.,k) \in {\rm C}^2(\re) \cap {\rm C}^\infty (\R^*)$, $j \in {\mathbb N}^*$, from \cite{HS1}[Proposition A1].
Moreover, the orthogonal projections $|\psi_{\eps,j}(k)\rangle \langle
\psi_{\eps,j}(k)|$ , $j \in {\mathbb N}^*$, depend analytically on $k$.

\subsection{Analysis of the band functions}\label{subsec:eps-band-fns1}

We treat the $\eps > 0$ problem as a perturbation of the $\eps = 0$ case. We first prove that the analysis of the band functions
$\omega_{\eps,j}(k)$ follows from that done for $\omega_j (k)$ in section \ref{sec:band-deriv1}.

\subsubsection{A comparison result for the band functions}

\begin{lemma}
\label{lm-b1}
Let $\eps>0$ and $j \in \mathbb{N}^*$. Then we have
\bel{b10}
| \omega_{\eps,j}(k) - \omega_j(k) | \leq  (r-1) (b_-^{1 \slash 2} \eps) \left(   (r-1) (b_-^{1 \slash 2} \eps)  + 2(2j-1)^{1 \slash 2} r^{1 \slash 2} \right) b_-,\ k \in \re.
\ee
\end{lemma}
\begin{proof}
Put $\mathfrak{a}:=\| \beta_\eps - \beta \|_{\infty}$. Since $h_{\eps}(k) - h(k)$ equals the difference of the potentials
$$ V_{\eps}(x,k)-V(x,k)=(\beta_{\eps}-\beta)(x)^2-2 (\beta_{\eps}-\beta)(x) \hat{v}_y(k), $$
and $\| \hat{v}_y(k) u \|^2 = \langle \hat{v}_y(k)^2 u , u \rangle \leq \langle h(k) u , u \rangle$, we have
\bel{b11}
\frac{|\langle ( h_{\eps}(k)-h(k)) u , u \rangle |}{\| u \|^2} \leq \mathfrak{a} \left( \mathfrak{a} +   2 \left(  \frac{\langle h(k) u , u \rangle}{\|u \|^2} \right)^{1 \slash 2} \right),
\ee
for all $u \in D(h(0)) \setminus \{ 0 \}$. Bearing in mind that $\mathfrak{a} \leq  (r-1)b_- \eps$ and
$\omega_j(k) \leq (2j-1) r b_-$, the result follows from \eqref{b11} and the minimax principle.
\end{proof}

\subsubsection{Positivity of the derivative of the band functions}

If we assume more regularity on $b(x)$ for $x \in [-\epsilon, \epsilon]$, we can prove a partial analog of Lemma \ref{lm-a1}.
However, this additional regularity is not needed for the main result, Theorem \ref{pr-b1}.

\begin{lemma}
\label{lm-bb1}
Fix $\eps>0$ and assume that $b_{\eps} \in {\rm C}^1(\re)$. Then we have
$$ \omega_{\eps,j}'(k) >0,\ j \in \N^*,\ |k| > b_+ \eps. $$
\end{lemma}

\begin{proof}
1.) By the Feynman-Hellmann theorem, we have
\beq\label{eq:band-der-epsilon1}
\omega_{\eps,j}'(k) = \left\langle \frac{d h_{\eps}}{d k}(k) \psi_{\eps,j}(x,k) , \psi_{\eps,j}(x,k) \right\rangle,
\eeq
where
\beq\label{eq:band-deriv-epsilon2}
\frac{dh_{\eps}}{d k} (k)=\frac{\partial  V_{\eps}}{\partial k}(x,k)= 2 (k-\beta_{\eps}(x)) = - \frac{1}{b_\eps (x)} \frac{\partial  V_{\eps}}{\partial x}(x,k).
\eeq
Hence, using the last formula on the right in \eqref{eq:band-deriv-epsilon2}, and integrating by parts, we obtain
\beas
\omega_{\eps,j}'(k) & = & -\int_{\re} \frac{\partial V_{\eps}}{\partial x} (x,k) \psi_{\eps,j}(x,k)^2 \frac{dx}{b_{\eps}(x)} \nonumber \\
& = & 2 \int_{\re} V_{\eps}(x,k)  \psi_j(x,k) \psi_j'(x,k) \frac{dx}{b_{\eps}(x)} - \int_{-\eps}^{\eps} V_{\eps}(x,k) \psi_j(x,k)^2 \frac{b_{\eps}'(x)}{b_{\eps}(x)^2}  dx.
\eeas
Putting this  and \eqref{b10} together we get that
$$
\omega_{\eps,j}'(k) = 2 \int_{\re} (\omega_{\eps,j}(k) \psi_{\eps,j}(x,k) + \psi_j''(x,k)) \psi_{\eps,j}'(x,k) \frac{dx}{b_{\eps}(x)}
- \int_{-\eps}^{\eps} V_{\eps}(x,k) \psi_{\eps,j}(x,k)^2 \frac{b_{\eps}'(x)}{b_{\eps}(x)^2} dx.
$$
The first term in the right hand side of the above identity reads
$$ \int_{\re} \frac{\partial}{\partial x} (\omega_{\eps,j}(k) \psi_{\eps,j}(x,k)^2 + \psi_{\eps,j}'(x,k)^2 ) \frac{dx}{b_{\eps}(x)} = \int_{-\eps}^{\eps} ( \omega_{\eps,j}(k) \psi_{\eps,j}(x,k)^2 + \psi_{\eps,j}'(x,k)^2 ) \frac{b_{\eps}'(x)}{b_{\eps}^2(x)} dx,
$$
hence
\bel{bb12}
\omega_{\eps,j}'(k) = \int_{-\eps}^{\eps} \left(  ( \omega_{\eps,j}(k) - V_{\eps}(x,k) ) \psi_{\eps,j}(x,k)^2 + \psi_{\eps,j}'(x,k)^2 \right) \frac{b_{\eps}'(x)}{b_{\eps}^2(x)}  dx.
\ee
Taking into account that $\frac{b_{\eps}'(x)}{b_{\eps}^2(x)}=-\frac{d}{dx} \left( \frac{1}{b_{\eps}(x)} \right)$ and integrating by parts in \eqref{bb12}, we get that
\bea
\omega_{\eps,j}'(k) 
& = & -\sum_{\zeta=+,-} \frac{\zeta}{b_{\zeta}} \left( ( \omega_{\eps,j}(k) - V_{\eps}(\zeta \eps,k) ) \psi_{\eps,j}(\zeta \eps,k)^2 + \psi_{\eps,j}'(\zeta \eps,k)^2 \right) \nonumber \\
& & + \int_{-\eps}^{\eps} \frac{\partial}{\partial x} \left( ( \omega_{\eps,j}(k) - V_{\eps}(x,k) ) \psi_{\eps,j}(x,k)^2 + \psi_{\eps,j}'(x,k)^2 \right) \frac{dx}{b_{\eps}(x)}. \label{bb12a}
\eea
In light of \eqref{b5} and \eqref{b10} we have
\beas
& & \frac{\partial}{\partial x} \left( ( \omega_{\eps,j}(k) - V_{\eps}(x,k) ) \psi_{\eps,j}(x,k)^2 + \psi_{\eps,j}'(x,k)^2 \right) \\
& = & -\frac{\partial V_{\eps}}{\partial x}(x,k) \psi_{\eps,j}(x,k)^2 = 2 ( k - \beta_{\eps}(x) ) \psi_{\eps,j}(x,k)^2 b_{\eps}(x),
\eeas
so \eqref{bb12a} entails
\bea
\omega_{\eps,j}'(k) & = & -\sum_{\zeta=+,-} \frac{\zeta}{b_{\zeta}} \left( ( \omega_{\eps,j}(k) - V_{\zeta}(\zeta \eps,k)^2 ) \psi_{\eps,j}(\zeta \eps,k)^2 + \psi_{\eps,j}'(\zeta \eps,k)^2 \right) \nonumber \\
& & + 2 \int_{-\eps}^{\eps} (k-\beta_{\eps}(x))  \psi_{\eps,j}(x,k)^2 dx. \label{bb12aa}
\eea

\noindent
2.) The next step involves relating $( \omega_{\eps,j}(k) - V_{\zeta}(\pm \eps,k)^2 ) \psi_{\eps,j}(\pm\eps,k)^2 + \psi_{\eps,j}'(\pm \eps,k)^2$ to $( \omega_{\eps,j}(k) - k)^2  \psi_{\eps,j}(0,k)^2 + \psi_{\eps,j}'(0,k)^2$. To this purpose we multiply the both sides of the following obvious identity
$$ \psi_{\eps,j}(\pm \eps,k)^2 = \psi_{\eps,j}(0,k)^2 + 2 \int_{0 \leq \pm x \leq \eps} \psi_{\eps,j}(x,k) \psi_{\eps,j}'(x,k) dx, $$
by $\omega_{\eps,j}(k)$, getting
\beas
\omega_{\eps,j}(k) \psi_{\eps,j}(\pm \eps,k)^2 & = & \omega_{\eps,j}(k) \psi_{\eps,j}(0,k)^2 \pm 2 \int_{0 \leq \pm x \leq \eps} ( h_{\eps} \psi_{\eps,j})(x,k) \psi_{\eps,j}'(x,k) dx \nonumber \\
& = & \omega_{\eps,j}(k) \psi_{\eps,j}(0,k)^2 \mp 2 \int_{0 \leq \pm x \leq \eps} ( \psi_{\eps,j}''(x,k) - V_{\eps}(x,k) \psi_{\eps,j}(x,k) ) \psi_{\eps,j}'(x,k) dx \nonumber \\
& = & \omega_{\eps,j}(k) \psi_{\eps,j}(0,k)^2 \mp \int_{0 \leq \pm x \leq \eps} \left( \frac{\partial \psi_{\eps,j}'}{\partial x}(x,k)^2 - V_{\eps}(x,k) \frac{\partial \psi_{\eps,j}}{\partial x}(x,k)^2 \right) dx .
\eeas
This yields
\beas
& & \psi_{\eps,j}'(\pm \eps,k)^2 + (\omega_{\eps,j}(k) - V_{\eps}(\pm \eps,k)) \psi_{\pm \eps,j}(\eps,k)^2 \nonumber \\
& =& \psi_{\eps,j}'(0,k)^2 + (\omega_{\eps,j}(k) - k^2) \psi_{\eps,j}(0,k)^2 \pm 2 \int_{0 \leq \pm x \leq \eps} (k - \beta_{\eps}(x)) \psi_{\eps,j}(x,k)^2 b_{\eps}(x) dx,
\eeas
by integrating by parts. From this and \eqref{bb12aa} it then follows that
\bea
\omega_{\eps,j}'(k) & = & \left( \frac{1}{b_-}-\frac{1}{b_+} \right) ( \psi_{\eps,j}'(0,k)^2 + (\omega_{\eps,j}(k) - k^2) \psi_{\eps,j}(0,k)^2 ) \\
& & + 2 \sum_{\zeta=+,-} \int_{0 \leq \zeta x \leq \eps} \left( 1 - \frac{b_{\eps}(x)}{b_{\zeta}} \right) (k-\beta_{\eps}(x))  \psi_{\eps,j}(x,k)^2 dx. \label{bb12aaa}
\eea
Further, by noticing that
\beas
\omega_{\eps,j}(k) \psi_{\eps,j}(0,k)^2  & = & \mp 2 \int_{0 < \pm x < \infty} h_{\eps}(k) \psi_{\eps,j}(x,k) \psi_{\eps,j}'(x,k) dx \\
& = & \pm \int_{0 < \pm x < \infty} \left( \frac{\partial \psi_{\eps,j}'}{\partial x}(x,k)^2 - V_{\eps}(x,k) \frac{\partial \psi_{\eps,j}}{\partial x}(x,k)^2 \right) dx \\
& = &  -  \psi_{\eps,j}'(0,k)^2 + k^2 \psi_{\eps,j}(0,k)^2 \mp \int_{0 < \pm x < \infty} (k - \beta_{\eps}(x)) \psi_{\eps,j}(x,k)^2  b_{\eps}(x) dx,
\eeas
we see that
$$ \psi_{\eps,j}'(0,k)^2+(\omega_{\eps,j}(k)-k^2) \psi_{\eps,j}(0,k)^2 = \pm 2 \int_{0 < \pm x < \infty} (k - \beta_{\eps}(x)) \psi_{\eps,j}(x,k)^2  b_{\eps}(x) dx. $$
This entails simultaneously
\bel{bb12b}
\omega_{\eps,j}'(k) = 2 \int_{-\infty}^{\eps} \left( 1 - \frac{b_{\eps}(x)}{b_+} \right)  (k-\beta_{\eps}(x)) \psi_{\eps,j}(x,k)^2  dx,
\ee
and
\bel{bb12c}
\omega_{\eps,j}'(k) =  -2 \int_{-\eps}^{+\infty} \left( \frac{b_{\eps}(x)}{b_-} - 1 \right)  (k-\beta_{\eps}(x)) \psi_{\eps,j}(x,k)^2  dx,
\ee
with the aid of \eqref{bb12aaa}.
The result now follows immediately from \eqref{bb12b} for $k > b_+ \eps$ and from \eqref{bb12c} for $k < -b_+ \eps$.
\end{proof}

\begin{remark}
Under the assumptions of Lemma \ref{lm-bb1} we deduce from \eqref{bb12} that
$$ \omega_{\eps,j}'(k)>0,\ | k | \leq b_+\eps,\ \eps \in \left( 0, (2j-1)^{1 \slash 2}  b_-^{-1 \slash 2} \slash (2 r) \right),\ j \in \N^*,$$
provided $b_\eps'(x) \geq 0$ for all $x \in (-\eps,\eps)$. Thus for every $\eps \in (0 ,b_-^{-1 \slash 2} \slash (2 r))$ we have
$$\omega_{\eps,j}'(k)>0,\ k \in \R,\ j \in \N^*, $$
under the above prescribed conditions on $b_\eps$. This result is similar to the one established in \cite{MP}[Remark 3.3] under slightly different hypothesis on the magnetic field.
\end{remark}

In light of Lemma \ref{lm-bb1} the band functions $k \mapsto \omega_{\eps,j}(k)$, $j \in \N^*$, are non constant for all $\eps>0$, thus the spectrum of $H_\eps$ is purely absolutely continuous. Moreover, we see from \eqref{a19} that
$\lim_{k \rightarrow \pm \infty} \| (h_\eps(k) - (2j-1) b_{\pm}) \psi_j^{\pm}(.,k) \| = 0$, hence
$$\lim_{k \rightarrow \pm \infty} \omega_{\eps,j}(k) = (2j-1) b_\pm,\ \eps >0, j \in \N^*,$$
by \eqref{b9}. As a consequence  we have
$$ \sigma(H_\eps)= \sigma_{ac}(H_\eps) = \bigcup_{j \in \N^*} \overline{\omega_{\eps,j}(\re)} =  \bigcup_{j \in \N^*} [ (2j-1) b_-,(2j-1) b_+],\ \eps>0. $$

\subsection{Existence of edge currents}

As in section \ref{sec:edge-currents1}, we define the current carried by a state $\varphi$ as the expectation of the $y$-component of the velocity operator \bel{bb13a}
v_{\eps,y}:=p_y-\beta_{\eps}(x)=v_y + (\beta-\beta_\eps)(x),
\ee
in the state $\varphi$, that is
\bel{bb13b}
J_{\eps,y}(\varphi):= \langle v_{\eps,y} \varphi , \varphi \rangle= J_y(\varphi) + \langle  (\beta-\beta_\eps) \varphi , \varphi \rangle.
\ee
The edge current does depend upon $\epsilon$ through $\beta_\epsilon$. However, formulae \eqref{bb13a}--\eqref{bb13b} show that the smooth Iwatsuka model may be treated as a perturbation of the sharp Iwatsuka model for small $\epsilon > 0$.

\subsubsection{Edge states carrying a current}

For all $j \geq 1$, we prove that edge currents exist for any $0 < \epsilon < \epsilon_j$ for energies in intervals $\Delta_j$. The existence of edge currents is related to the existence of absolutely continuous spectrum. As mentioned after Lemma \ref{lm-a1}, as long as the band functions are non constant, the spectrum is absolutely continuous. We established this for the smooth Iwatsuka model $H_\epsilon$ in two cases. First, it follows from Lemma \ref{lm-b1} and \eqref{b10} that if $b_-^{1/2} \epsilon << 1$, then the two band functions $\omega_{\epsilon,j}(k)$ and $\omega_j(k)$ are uniformly close. Since $\omega_j(k)$ is monotone increasing by Lemma \ref{a1}, the band function $\omega_{\epsilon,j}(K)$ cannot be constant. Second, if we suppose that $b_\epsilon \in C^1 (\R)$, it follows from the above remark and Lemma \ref{lm-bb1} that the band functions are non constant with no constraint on $\epsilon$. We mention that Iwatsuka \cite{iwatsuka1} proves absolutely continuity of the spectrum provided the magnetic field $b(x)$ is smooth $b(x) \in C^\infty (\R)$, it is bounded $0 < M_- \leq b(x) \leq M_+ < \infty$, and $\limsup_{x \rightarrow - \infty} b(x)< \liminf_{x \rightarrow \infty} b(x)$ or the reverse inequality. Furthermore, under the additional condition that $b(x)$ is monotone (without any regularity assumption), Dombrowski, Germinet, and Raikov \cite[Corollary 2.3]{DGR} proved the quantization of the edge current (see section \ref{subsec:relation1}).

\begin{theorem}
\label{pr-b1}
Let $b_-,r, n, j,\delta_j$ and $\Delta_j$ be as in Proposition \ref{prop-a2}. Then there exists $\eps_j>0$, depending on $b_-$, such that for each $\eps \in (0, \eps_j)$, we may find a subinterval $\Delta_{\eps,j}$ of $\Delta_j$, with same midpoint $E_j$, satisfying
\beq\label{eq:current-pert1}
 J_{\eps,y}(\psi) \geq \frac{c_j}{2} \delta_j^3 \left( \frac{r-1}{r^3} \right) b_-^{1 \slash 2} \| \psi \|^2,\ \psi=\mathbb{P}_{\eps}(\Delta) \psi,
 \eeq
for any subinterval $\Delta \subset \Delta_{\eps,j}$ centered at $E_j$.
Here $\mathbb{P}_{\eps}(I)$ denotes the spectral projection of $H_{\eps}$ for the Borel set $I \subset \re$ and the constant $c_j>0$ is defined by \eqref{a14}.
\end{theorem}

\begin{proof}
1. We perform a decomposition of $\psi$ in order to calculate the current.
We set $d_j=|\Delta_j | \slash (2b_-)=(r+1)((2j-1) (r-1) \slash (r+1) - \delta_j)  \slash 2$ and, for $N \geq 1$, consider the subinterval $\Delta_{j,N} = (E_j- d_{j,N} b_-, E_j+ d_{j,N}b_-)$ of $\Delta_j$, with $d_{j,N}:=d_j \slash N$. Then we
decompose $\psi=\mathbb{P}_{\eps}(\Delta_{j,N}) \psi$, into the sum
\bel{bb20}
\psi=\phi+\xi,\ \phi:=\mathbb{P}(\Delta_j) \psi,\ \xi:=\mathbb{P}(\Delta_j^c) \psi,
\ee
where $\Delta_j^c:= \re \setminus \Delta_j$.

\noindent
2. We next estimate the perturbation. Since $W_\eps:=H_\eps - H = - 2 (\beta_\eps - \beta) v_y + (\beta_\eps-\beta)^2$ and $\| v_y \psi \|  = \langle v_y^2 \psi , \psi \rangle^{1 \slash 2} \leq  \langle H \psi , \psi \rangle^{1 \slash 2} \leq  \| H \psi \|^{1 \slash 2} \| \psi \|^{1 \slash 2}
\leq ( \| H_\eps \psi \|^{1 \slash 2} + \| W_\eps \psi \|^{1 \slash 2} ) \| \psi \|^{1 \slash 2}$, we have
\bel{bb20a}
\frac{\| W_\eps \psi \|}{\| \psi \|} \leq \|\beta_\eps - \beta \|_{\infty} \left(  \|\beta_\eps - \beta \|_{\infty}  + 2  \left( \frac{\| H_\eps \psi \|}{\| \psi \|} \right)^{1 \slash 2}+ 2 \left(  \frac{\| W_\eps \psi \|}{\|\psi \|} \right)^{1 \slash 2} \right).
\ee
Bearing in mind that $\| \beta_\eps - \beta \|_{\infty} \leq  \mathfrak{a}  b_-^{1 \slash 2}$, with $\mathfrak{a}:=(r-1) (b_-^{1 \slash 2} \eps)$, and that $\| H_\eps \psi \| \leq (e_j+ d_{j,N}) b_- \| \psi \|$,
where we have set $e_j:=E_j \slash b_-$, it follows from \eqref{bb20a} that $t= \| W_\eps \psi \| \slash ( b_- \| \psi \|)$ is
a solution to the inequality
$$
t \leq \mathfrak{a}  \left( \mathfrak{a}  + 2 ( e_j + d_{j,N})^{1 \slash 2}  +  2 t^{1 \slash 2} \right).
$$
As a consequence, we have $\| W_{\eps} \psi \| \leq 2 \mathfrak{a} ( 2 \mathfrak{a}^{1 \slash 2} +( e_j + d_{j,N})^{1 \slash 4})^2 b_- \| \psi \|$,
which implies that
\bel{b24}
\| (H - E_j) \psi \| \leq \| W_\eps \psi \|+ \| (H_\eps -E_j) \psi \| \leq c_{j,N}(\mathfrak{a}) b_-\| \psi \|,
\ee
where
\bel{b24b}
c_{j,N}(\mathfrak{a}):=2 \mathfrak{a} ( 2 \mathfrak{a}^{1 \slash 2} +( e_j + d_{j,N})^{1 \slash 4})^2 +d_{j,N}.
\ee

\noindent
3. We next estimate $\| \xi\|$ and $\| \phi \|$. From \eqref{b24} and the definition of $\xi$, we have
\bel{b22}
\| \xi \| \leq \tilde{c}_{j,N}(\mathfrak{a}) \| \psi \| ,\ \tilde{c}_{j,N}:=c_{j,N}(\mathfrak{a}) \slash d_j,
\ee
since $\xi = \mathbb{P}(\Delta_j^c) (H - E_j)^{-1} (H - E_j) \psi$ and  $\| \mathbb{P}(\Delta_j^c) (H - E_j)^{-1} \| \leq 1 \slash (d_j b_-)$ as a bounded operator in $L^2(\re^2)$.
Further, $\phi$ and $\xi$ being orthogonal in $L^2(\re^2)$, we deduce from \eqref{b22} that
\bel{b22b}
\| \phi \|^2 = \| \psi \|^2 - \| \xi \|^2 \geq  ( 1 -  \tilde{c}_{j,N}(\mathfrak{a})^2 )  \| \psi \|^2.
\ee

\noindent
4. Applying these estimates to the current, we recall from \eqref{bb13a}-\eqref{bb13b} that
\bel{b23a}
J_{y,\eps}(\psi)  \geq J_{y}(\psi) - \mathfrak{a} b_-^{1 \slash 2} \| \psi \|^2,
\ee
and from \eqref{bb20} that
\bel{b24a}
J_{y}(\psi)  =  J_{y}(\phi)+ 2 \Pre{
\langle v_{y} \xi , \phi \rangle} + \langle v_{y} \xi, \xi \rangle \geq J_{y}(\phi) - \rho(\phi,\xi),
\ee
where
\bel{b25a}
\rho(\phi,\xi):=2 | \langle v_{y} \xi , \phi \rangle | + | \langle v_{y} \xi, \xi \rangle | \leq 3 \| v_{y} \xi  \| \| \psi \|.
\ee
Here we used once more the orthogonality of $\phi$ and $\xi$ in $L^2(\re^2)$.
Next, by applying Theorem \ref{lm-a3} and using \eqref{b22b}, we bound from below the first term in the right hand side of \eqref{b24a} as
\bel{b26a}
J_{y}(\phi) \geq c_j \delta_j^3  \left( \frac{r-1}{r^3} \right)  ( 1 -   \tilde{c}_{j,N}(\mathfrak{a})^2 )  b_-^{1 \slash 2}  \| \psi \|^2.
\ee

\noindent
5. The next step of the proof is to improve the upper bound  \eqref{b25a} on the remaining term $\rho(\phi,\xi)$. This can be achieved by noticing that
$$ \| v_{y} \xi \|^2 = \langle v_{y}^2 \xi , \xi \rangle \leq \langle H \xi , \xi \rangle \leq \langle H \xi , \psi \rangle \leq \langle  \xi , H \psi \rangle \leq \| \xi \| \|H \psi \|, $$
since
$\langle H \xi , \phi \rangle = \langle H \mathbb{P} (\Delta_j^c) \psi ,  \mathbb{P}(\Delta_j) \psi \rangle = \langle  \mathbb{P}(\Delta_j^c) H \psi ,  \mathbb{P}(\Delta_j) \psi \rangle = 0$,
and combining \eqref{b22} with the estimate $\|H \psi \| \leq  E_j \| \psi \| + \| (H-E_j)\psi \| \leq (e_j+ c_{j,N}(\mathfrak{a})) b_- \| \psi \|$
arising from \eqref{b24}-\eqref{b24b}. We get that
\bel{b27a}
\rho(\phi,\xi) \leq 3  \tilde{c}_{j,N}(\mathfrak{a})^{1 \slash 2} ( e_j+ c_{j,N}(\mathfrak{a}))^{1 \slash 2} b_-^{1 \slash 2} \| \psi \|^2.
\ee
Finally, putting \eqref{b23a}-\eqref{b24a} and \eqref{b26a}-\eqref{b27a} together, we end up getting
$$ J_{y,\eps}(\psi)  \geq F_{j,N}(\mathfrak{a}) b_-^{1 \slash 2} \|\psi \|^2,$$
where
$$ F_{j,N}(\mathfrak{a}) := c_j  \delta_j^3 \left( \frac{r-1}{r^3} \right)  (1- \tilde{c}_{j,N}(\mathfrak{a})^2)  - \left( 3  \tilde{c}_{j,N}(\mathfrak{a})^{1 \slash 2} (e_j+ c_{j,N} (\mathfrak{a}))^{1 \slash 2}+  \mathfrak{a} \right).$$
Finally, we take $N$ sufficiently large and $\mathfrak{a}$ sufficiently small so that
\beq\label{eq:condition-local1}
\tilde{c}_{j,N}(\mathfrak{a})^2+c_j^{-1} \delta_j^{-3} \left( \frac{r^3}{r-1} \right) \left( 3 \tilde{c}_{j,N}(\mathfrak{a})^{1 \slash 2}(e_j+ c_{j,N} (\mathfrak{a}))^{1 \slash 2}+  \mathfrak{a} \right) \leq \frac{1}{2}.
\eeq
This gives the lower bound \eqref{eq:current-pert1}. Note that $a = (r-1)(b_-^{-1/2} \epsilon)$ so condition
\eqref{eq:condition-local1} requires that $\epsilon < b_-^{-1/2}$.
\end{proof}

\subsubsection{Localization of the edge currents}

Continuing to consider $\beta_\epsilon$ as a perturbation of $\beta$, we are able to prove that the perturbed edge currents remain localized in a small neighborhood of $x=0$ under the hypothesis of Theorem \ref{pr-b1} that $\epsilon$ is small relative to $b_-^{-1/2}$.

\begin{theorem}
\label{pr-b2}
Let $b_-$, $r$, $n$, $j$, $\eps_j>0$ and $\Delta_{\eps,j}$, for some $\eps \in (0,\eps_j)$, be the same as in Theorem \ref{pr-b1}. Then for all $\varepsilon_1>0$ and $\varepsilon_2>0$ there
exists $b_j(\varepsilon_1,\varepsilon_2)>0$ such that any ${\rm L}^2(\re^2)$-normalized function $\psi = {\mathbb P}_{\eps}(\Delta) \psi$, where $\Delta$ is any subinterval of $\Delta_{\eps,j}$, satisfies
$$ \int_{\re^2} \chi_{I_{\varepsilon_1,\varepsilon_2}}(x) | \psi(x,y) |^2 {\rm d} x {\rm d} y
\geq  1 - \eta_j {\rm e}^{-\varepsilon_1^2 b_-^{2 \varepsilon_2} \slash 8}, $$
provided $b_- \geq b_j(\varepsilon_1,\varepsilon_2)$.
Here $\chi_{I_{\varepsilon_1,\varepsilon_2}}$ denotes the characteristic function of the interval $I_{\varepsilon_1,\varepsilon_2}:=
[-(1+\varepsilon_1) b_-^{-1 \slash 2+\varepsilon_2},(1+\varepsilon_1) b_+^{-1 \slash 2 + \varepsilon_2}]$, and
$\eta_j = 2 ( \pi  (2j-1))^{1/2}$, both as in Theorem \ref{lm-a4}.
\end{theorem}

\begin{proof}
The proof is similar to the one of Theorem \ref{lm-a4}.
Setting $k_{\varepsilon_2}=-b_-^{1\slash 2 + \varepsilon_2}$ and arguing in the exact same way as
Step 1 in the proof of Theorem \ref{lm-a4} we find some $b_{j}(\varepsilon_2)>0$ such that ${\rm dist}(\sigma(h(k_{\varepsilon_2})),(2j-1)b_-) < (\delta_j \slash 2) b_-$, and thus
$(2j-1) b_- < \omega_j(k_{\varepsilon_2}) < (2j-1)b_- + (\delta_j / 2) b_-$,
for every $b \geq b_{j}(\varepsilon_2)$. This yields
\bel{b30}
b_- \geq b_{j}(\varepsilon_2) \Longrightarrow (2j-1) b_- < \omega_j(k) < (2j-1)b_- + \frac{\delta_j}{2} b_-,\ k \leq k_{\varepsilon_2},
\ee
by Lemma \ref{lm-a1}. Further, we choose $\tilde{\eps}_{j}=\tilde{\eps}_{j}(b_-)>0$ so small that the right hand side of
\eqref{b10b}, where $\left( \frac{2n+1}{2n-1} \right)^{1 \slash 2}$ and $\tilde{\eps}_{j}$ are respectively substituted for $r$ and $\eps$, is smaller than $(\delta_j \slash 2) b_-$. Then, due to \eqref{a12d} and Lemma \ref{lm-b1}, we deduce from
\eqref{b30} that
$$ b_- \geq b_{j}(\varepsilon_2) \Longrightarrow \omega_{\eps,j}(k) < (2j-1)b_- + \delta_j b_-,\ k \leq k_{\varepsilon_2},\ \eps \in (0,\tilde{\eps}_{j}),$$
hence
$$ b_- \geq b_{j}(\varepsilon_2) \Longrightarrow \omega_{\eps,j}^{-1}(\Delta_j) \subset (k_{\varepsilon_2},+\infty),\
\eps \in (0,\tilde{\eps}_{j}). $$
Now, doing the same with $k_{\varepsilon_2}=b_+^{1 \slash 2 + \varepsilon_2}$ we end up getting some constant $b_{j}(\varepsilon_2)>0$ such that
$$
b_- \geq b_{j}(\varepsilon_2) \Rightarrow ( \exists \tilde{\eps}_j=\tilde{\eps}_j(b_-)>0,\
\omega_{\eps,j}^{-1}(\Delta_j) \subset (-b_-^{1 \slash 2 + \varepsilon_2}, b_+^{1 \slash 2 + \varepsilon_2}),\ \eps \in (0,\tilde{\eps}_{j}) ).
$$
For the sake of simplicity, let us denote $\min( \eps_j ,\tilde{\eps}_j)$, where $\eps_j$ is the same as in Theorem \ref{pr-b1}, by $\eps_j$. Then, $\Delta$ being a subset of $\Delta_j$ for each $\eps \in (0,\eps_j)$, it follows readily from the above implication that
\bel{b32}
b_- \geq b_{j}(\varepsilon_2) \Longrightarrow
\omega_{\eps,j}^{-1}(\Delta_{\eps,j} \subset (-b_-^{1 \slash 2 + \varepsilon_2}, b_+^{1 \slash 2 + \varepsilon_2}),\ \eps \in (0,\eps_{j}).
\ee
Let us now fix $b_- \geq b_{j}(\varepsilon_2)$ and $\eps \in (0,\eps_{j} \slash b_-^{1 \slash 2})$. We notice that
$\pm x_j^{\pm}(\varepsilon_2) > \eps$, where $x_j^{\pm}$ is defined by \eqref{a31}, so we have
$V_{\eps}(x,k)=V(x,k)=(k-b_{\pm} x)^2$ for every $\pm x \geq \pm x_j^{\pm}$. From this and \eqref{b32} then follows for each $\eps \in (0,\eps_j)$ that
$$
|\psi_{\eps,j}(x,k)| \leq 2^{1 \slash 2} (2j-1)^{1 \slash 4} b_+^{1 \slash 4} {\rm e}^{-b_{\pm}(x-x_j^{\pm}(\varepsilon_2))^2 \slash 2},\ \pm x \geq \pm x_j^{\pm}(\varepsilon_2),\ k \in \omega_{\eps,j}^{-1}(\Delta_{\eps,j}),
$$
by just mimicking Step 2 in the proof of Theorem \ref{lm-a4}.

Having said that, the end of the proof now applies without change upon substituting $\psi$ (resp. $\omega_{\eps,j}$, $\psi_{\eps,j}$) for $\varphi$ (resp. $\omega_j$, $\psi_j$).
\end{proof}


\section{Perturbations of Iwatsuka Hamiltonians: Stability of edge currents}\label{sec:pert1}
\setcounter{equation}{0}

We now consider the perturbation of $H_\eps=H(A_\eps)$ defined in \eqref{b1}-\eqref{b2} by a magnetic potential
$a (x,y)=(a_1(x,y),a_2(x,y)) \in {\rm W}^{1,\infty}(\re^2)$ and a bounded scalar potential $q(x,y) \in {\rm L}^{\infty}(\re^2)$.
We prove that the lower bound on the edge current in Theorem \ref{pr-b1} is stable with respect to
these perturbations provided $\| a \|_{\infty}$ and $\|q \|_\infty$ are not too large
compared with $b_\pm$ in a sense to be made precise.

Prior to establishing this result we introduce some useful notation and
rigorously define the perturbed Hamiltonian under study. To this end we introduce
\bel{emp1}
W_\eps(a):=H(A_\eps+a)-H_\eps= - 2 a \cdot (i \nabla + A_\eps) - i (\nabla \cdot a) +  | a |^2,
\ee
where $ | a |^2$ is a shorthand for $a \cdot a$.
Since
$\|  (-i \nabla - A_\eps) \varphi \| = \langle H_\eps \varphi , \varphi \rangle^{1 \slash 2} \leq \lambda \| H_\eps \varphi  \| + \lambda^{-1} \| \varphi  \|$ for all $\varphi \in C_0^{\infty}(\re^2)$  and $ \lambda >0$, we have that
$$ \| W_\eps(a) \varphi \| \leq 2 \lambda \| a \|_{\infty} \| H_\eps \varphi \| + (\lambda^{-1} +  \|\nabla a \|_\infty + \| a \|_{\infty}^2) \| \varphi \|,\ \lambda >0,$$
by \eqref{emp1}. This guarantees that $W_\eps(a)$ is $H_\eps$-bounded with relative bound smaller than one provided $\lambda \in (0 , 1 \slash (2 \| a \|_{\infty}))$, so the operator $H(A_\eps+a)$ is selfadjoint in ${\rm L}^2(\re^2)$, with same domain as $H_\eps$ from \cite{RS}[Theorem X.12].  Moreover the same is true for $H(A_\eps+a,q):=H(A_\eps+a)+q$
since $q \in {\rm L}^{\infty}(\re^2)$.

Following the ideas of \S \ref{sec:edge-currents1} and \S \ref{sec:smooth1} we may now define the second component of the velocity operator associated to $H(A_\eps+a,q)$ as
\bel{emp1b}
v_{y,A_\eps +a,q}=v_{y,A_\eps+a} := \frac{i}{2} [ H(A_\eps+a,q),y]=v_{y,\eps} - a_2 ,
\ee
and the current carried by a quantum state $\psi$ as
\bel{emp1c}
J_{y,A_\eps+a,q}(\psi) := \langle v_{y, A_\eps+a} \psi, \psi \rangle.
\ee
Notice that we keep the subscript $q$ in the left hand side
of the identity \eqref{emp1c} although the $y$-component of the velocity is independent of $q$ according to \eqref{emp1b}, and $q$ is nowhere to be seen in the right hand side of \eqref{emp1c}. This is actually justified by the fact that the state $\psi$ we shall consider in the sequel is determined from $H(A_\eps+a,q)$ and thus depends on $q$, along with the current it carries.

\begin{theorem}
\label{thm-emp1}
Let $b_-$, $r$, $n$, $j$, $\delta_j$, $\eps_j$ and $\Delta_{\eps,j}$, for some fixed $\eps \in (0,\eps_j)$,  be the same as
in Theorem \ref{pr-b1}. Then there are three constants $a_*>0$, $q_*>0$ and
$d_*>0$, all of them being independent of $b_-$, such that for all $a \in {\rm W}^{1,\infty}(\re^2)$ obeying $(\| a \|_{\infty}^2 + \| \nabla a \|_{\infty})^{1 \slash 2} \leq a_* b_-^{1/2}$ and all $q \in {\rm L}^{\infty}(\re^2)$ with $\| q \|_{\infty} \leq q_* b_-$, the following estimate
\bel{emp1d}
J_{y,A_\eps+a,q}(\psi) \geq \frac{c_j}{4} \delta_j^3 \left( \frac{r-1}{r^3} \right) b_-^{1 \slash 2} \| \psi \|^2,\ \psi \in  \mathbb{P}_{A_\eps+a,q}(\Delta) {\rm L}^2(\re^2),
\ee
holds true for any subinterval $\Delta$ of $\Delta_{\eps,j}$ with same midpoint, satisfying $| \Delta| \leq d_* b_-$.
Here $\mathbb{P}_{A_\eps+a,q}(I)$ is the spectral projection of $H(A_\eps+a,q)$ for the Borel set $I \subset \re$, and $c_j$ is the constant introduced in Theorem
\ref{lm-a3}.
\end{theorem}
\noindent
{\bf Proof.}
The proof is similar to the one of Theorem \ref{pr-b1}.
Put $d_{j}:=| \Delta_{\eps,j} |\slash (2 b_-)$ and $d_{j,N}:=d_{j} \slash N$, for some $N \geq 1$. Further, introduce the set $\Delta_{j,N}:=(E_j - d_{j,N} b_-,E_j + d_{j,N} b_-)$, where $E_j$ is the center of $\Delta_{\eps,j}$, and decompose $\psi = \mathbb{P}_{A_\eps+a,q}(\Delta_{j,N}) \psi$ into the sum
\bel{emp2}
\psi = \phi + \xi,\ \phi:=\mathbb{P}_\eps(\Delta_{\eps,j}) \psi,\ \xi:=  \mathbb{P}_\eps(\Delta_{\eps,j}^c) \psi,
\ee
where $\Delta_{\eps,j}^c=\re \setminus \Delta_{\eps,j}$.
Since $ \| (i \nabla +A_\eps) \psi \|  = \langle H_\eps \psi , \psi \rangle^{1 \slash 2} \leq \| H_\eps \psi \|^{1 \slash 2} \|\psi \|^{1 \slash 2}$ and
$$\| H_\eps \psi \| \leq \| H(A_\eps+a,q) \psi \| + \|q \|_{\infty} \| \psi \| + \|W_\eps(a) \psi \|
\leq  (e_j+d_{j,N}+\mathfrak{q}) b_- \| \psi \| + \|W_\eps(a) \psi \|, $$
with $e_j:=E_j \slash b_-$ and $\mathfrak{q}:= \| q \|_\infty \slash b_-$,
we deduce from \eqref{emp1} that
$$ \frac{\| W_\eps(a) \psi \|}{b_- \| \psi \|} \leq \mathfrak{a} \left(  \mathfrak{a} + 2 (e_j+d_{j,N}+\mathfrak{q})^{1 \slash 2} + 2 \left( \frac{\| W_\eps(a) \psi \|}{b_- \| \psi \|} \right)^{1 \slash 2} \right), $$
where $\mathfrak{a}:= (\| a \|_\infty^2 + \| \nabla a \|_{\infty} )^{1 \slash 2} \slash b_-^{1 \slash 2}$.
This entails
$\| W_\eps(a) \psi \|  \leq  2 \mathfrak{a} ( 2 \mathfrak{a}^{1 \slash 2} +  (e_j + d_{j,N} + \mathfrak{q} )^{1 \slash 4} )^2 b_- \| \psi \|$
and thus
\bel{emp3a}
\| (H_\eps - E_j) \psi \| \leq \| (W_\eps(a)+q) \psi \|+ \| (H(A_\eps + a,q) -E_j) \psi \| \leq  c_{j,N}(\mathfrak{a},\mathfrak{q}) b_- \| \psi \|,
\ee
where
\bel{emp3b}
c_{j,N}(\mathfrak{a},\mathfrak{q}):=2 \mathfrak{a} ( 2 \mathfrak{a}^{1 \slash 2} + (e_j + d_{j,N} + \mathfrak{q} )^{1 \slash 4} )^2+ d_{j,N} + \mathfrak{q}.
\ee
As a consequence  we have
\bel{emp5}
\| \xi \| \leq \tilde{c}_{j,N}(\mathfrak{a},\mathfrak{q}) \| \psi \|,\ \tilde{c}_{j,N}(\mathfrak{a},\mathfrak{q}):= c_{j,N}(\mathfrak{a},\mathfrak{q}) \slash d_j,
\ee
since $\xi = \mathbb{P}_\eps(\Delta_{\eps,j}^c) (H_\eps - E_j)^{-1} (H_\eps - E_j) \psi$ and  $\| \mathbb{P}_\eps(\Delta_{\eps,j}^c) (H_\eps - E_j)^{-1} \| \leq 1 \slash (d_j b_-)$ as a bounded operator in $L^2(\re^2)$.
Moreover, $\phi$ and $\xi$ being orthogonal in $L^2(\re^2)$, it follows from \eqref{emp5} that
\bel{emp5a}
\| \phi \|^2 = \| \psi \|^2 - \| \xi \|^2 \geq ( 1 -  \tilde{c}_{j,N}(\mathfrak{a},\mathfrak{q})^2)  \| \psi \|^2.
\ee
Now recall from \eqref{emp1b}-\eqref{emp1c} that
\bel{emp3}
J_{y,A_\eps+a,q}(\psi) = J_{y,\eps}(\psi)- \langle a_2 \psi , \psi \rangle \geq J_{y,\eps}(\psi) - \mathfrak{a} b_-^{1 \slash 2} \| \psi \|^2,
\ee
and from \eqref{emp2} that
\bel{emp4}
J_{y,\eps}(\psi)  =  J_{y,\eps}(\phi)+ 2 \Pre{
\langle v_{y,\eps} \xi , \phi \rangle} + \langle v_{y,\eps} \xi, \xi \rangle \geq J_{y,\eps}(\phi) - \rho(\phi,\xi),
\ee
where
\bel{emp4d}
\rho(\phi,\xi):=2 | \langle v_{y,\eps} \xi , \phi \rangle | + | \langle v_{y,\eps} \xi, \xi \rangle | \leq 3 \| v_{y,\eps} \xi  \| \| \psi \|.
\ee
Here we used once more the orthogonality of $\phi$ and $\xi$ in $L^2(\re^2)$.

Further, applying Theorem \ref{pr-b1} and using \eqref{emp5a}, the first term in the right hand side of \eqref{emp4} is bounded from below as
\bel{emp4c}
J_{y,\eps}(\phi) \geq \frac{c_j}{2} \delta_j^3 \left( \frac{r-1}{r^3} \right)  ( 1 -  \tilde{c}_{j,N}(\mathfrak{a},\mathfrak{q})^2 ) b_-^{1 \slash 2} \| \psi \|^2.
\ee
The next step of the proof is to improve the upper bound  \eqref{emp4d} on the remaining term $\rho(\phi,\xi)$. This can be achieved by noticing that
$$ \| v_{y,\eps} \xi \|^2 = \langle v_{y,\eps}^2 \xi , \xi \rangle \leq \langle H_\eps \xi , \xi \rangle \leq \langle H_\eps \xi , \psi \rangle \leq \langle  \xi , H_{\eps} \psi \rangle \leq \| \xi \| \|H_\eps \psi \|, $$
since
$\langle H_\eps \xi , \phi \rangle = \langle H_\eps \mathbb{P}_\eps (\Delta_{\eps,j}^c) \psi ,  \mathbb{P}_\eps(\Delta_{\eps,j}) \psi \rangle = \langle  \mathbb{P}_\eps(\Delta_{\eps,j}^c) H_\eps \psi ,  \mathbb{P}_\eps(\Delta_{\eps,j}) \psi \rangle = 0$,
and combining \eqref{emp5} with the estimate $\|H_\eps \psi \| \leq E_j \|\psi \| + \| (H_\eps-E_j) \psi \| \leq (e_j+c_{j,N}(\mathfrak{a},\mathfrak{q}) )  b_-\| \psi \|$
arising from \eqref{emp3a}-\eqref{emp3b}. We get
$\| v_{y,\eps} \xi \| \leq \| H_\eps \psi \|^{1 \slash 2} \| \xi \|^{1 \slash 2} \leq \tilde{c}_{j,N}(\mathfrak{a},\mathfrak{q})^{1 \slash 2} (e_j+c_{j,N}(\mathfrak{a},\mathfrak{q}))^{1 \slash 2} b_-^{1 \slash 2}\| \psi \|$, hence
\bel{emp8}
\rho(\phi,\xi) \leq 3 \tilde{c}_{j,N}(\mathfrak{a},\mathfrak{q})^{1 \slash 2} (e_j+c_{j,N}(\mathfrak{a},\mathfrak{q}))^{1 \slash 2} b_-^{1 \slash 2} \| \psi \|^2,
\ee
by \eqref{emp4d}.
Finally, putting \eqref{emp3}-\eqref{emp4} and \eqref{emp4c}-\eqref{emp8} together, we end up getting that
$$ J_{y,A_\eps+a,q} \geq F_{j,N}(\mathfrak{a},\mathfrak{q})b_-^{1 \slash 2} \|\psi \|^2,$$
where
$$ F_{j,N}(\mathfrak{a},\mathfrak{q}):= \frac{c_j}{2}\delta_j^3 \left( \frac{r-1}{r^3} \right) ( 1 -  \tilde{c}_{j,N}(\mathfrak{a},\mathfrak{q})^2 )   - 3 \tilde{c}_{j,N}(\mathfrak{a},\mathfrak{q})^{1 \slash 2} (e_j+c_{j,N}(\mathfrak{a},\mathfrak{q}))^{1 \slash 2} - \mathfrak{a}. $$
Last, taking $1 \slash N$, $\mathfrak{a}$ and $\mathfrak{q}$ so small that
$$ \tilde{c}_{j,N}(\mathfrak{a},\mathfrak{q})^2   + 2 c_j^{-1} \delta_j^{-3} \left( \frac{r^3}{r-1}  \right) (3 \tilde{c}_{j,N}(\mathfrak{a},\mathfrak{q})^{1 \slash 2} (e_j+c_{j,N}(\mathfrak{a},\mathfrak{q}))^{1 \slash 2} +  \mathfrak{a} )\leq \frac{1}{2},
$$
we obtain the desired result.
$\Box$

In light of \eqref{emp1b}-\eqref{emp1c}, the inequality  \eqref{emp1d}, may be equivalently rephrased as the following Mourre estimate
\bel{emp7}
 \mathbb{P}_{A_\eps+a,q}(\Delta) i [H(A_\eps+a,q), y ] \mathbb{P}_{A_\eps+a,q}(\Delta) \geq \frac{c_j}{4} \delta_j^3 \left( \frac{r-1}{r^3} \right) b_-^{1 \slash 2} \mathbb{P}_{A_\eps+a,q}(\Delta).
\ee
Moreover $y$ is a bona-fide conjugate operator for the magnetic operator $H(A_\eps+a,q)$ in the sense of Mourre since $(i/2)[H(A_\eps+a,q),y]=(v_{\eps,y} - a_2)$ is bounded from the domain of $H$ to $L^2(\re^2) $, and  the double commutator $i[i[H(A_\eps+a,q), y ],y]=2$. Therefore, in view of \eqref{emp7} and \cite{CFKS}[Corollary 4.10] we have obtained the:

\begin{follow}
\label{cor-acs}
Under the conditions, and with notations, of Theorem \ref{thm-emp1}, the spectrum of the operator $H(A_\eps+a,q)$, $\eps \in (0,\eps_j)$,  is purely absolutely continuous in the interval $\Delta$:
$$\sigma(H(A_\eps+a,q)) \cap \Delta = \sigma_{ac}(H(A_\eps+a,q)) \cap \Delta.$$
\end{follow}
The existence of edge currents for energies in a suitable subinterval of $\R_+$ is thus equivalent to the existence of purely absolutely continuous spectrum
for $H(A_\eps+a,q) $ in the corresponding interval.


\section{Persistence of edge currents in time: Asymptotic velocity}\label{sec:time-behavior1}
\setcounter{equation}{0}

We investigate the time evolution of the current under the unitary evolution groups generated by the Iwatsuka and perturbed Iwatsuka Hamiltonians. The general situation we address is the following. Let $H$ be a self-adjoint Schr\"odinger operator on $L^2 (\R^2)$.
          This operator generates the unitary time evolution group $U(t) = e^{-itH}$. Let $v_y := (i/2) [ H, y]$ be the $y$-component of the velocity operator. We are interested in evaluating the asymptotic time behavior of $\langle U(t) \varphi, v_y U(t) \varphi \rangle$ as
$t \rightarrow \pm \infty$.



The lower bounds on the edge currents for the two unperturbed models, the sharp and smooth Iwatsuka models, are valid for all times.
It we replace $v_y$ in \eqref{eq:lower-bound1} by $v_y (t) := e^{itH}v_ye^{-itH}$, then the lower bound remains valid since the state $\varphi (t) := U(t) \varphi$ satisfies $\mathbb{P} (\Delta_j) \varphi(t) = \varphi (t)$ for all time.
Similarly, if we replace $v_{\epsilon,y}$ in \eqref{eq:current-pert1} by its time evolved current $v_{\epsilon,y}(t)$
using the operator $U_\epsilon (t) = e^{-itH_\epsilon}$, then the lower bound in \eqref{eq:current-pert1} remains valid for all time.

Perturbed Hamiltonians $H(A_\epsilon+a,q)$ were treated in section \ref{sec:pert1}. Theorem \ref{thm-emp1} states that if the ${\rm L}^\infty$-norms of $a_j^2$, $\nabla a_j$, for $j=1,2$, and of $q$ are small relative to $b_-$, then the edge current $J_{y, A_\epsilon +a,q}(\psi)$ is bounded from below for all $\psi \in \mathbb{P}_{A_\epsilon+a,q} (\Delta) {\rm L}^2(\R^2)$, where $\Delta \subset \Delta_{\epsilon,j}$, with $\Delta_{\epsilon,j}$ as defined in Theorem \ref{pr-b1}. By the same reasoning as above, the same lower bound holds for the time-evolved edge current $\langle \psi, v_{\epsilon, A_\epsilon+a,q} (t) \rangle$ for all time.
The boundedness of $a_j$, $\nabla a_j$, and of $q$ is rather restrictive. From the form of the current operator in \eqref{emp1b}, it would appear that only $\|a_2\|_\infty$ needs to be controlled. We prove here that if we limit the support of the perturbation $(a_1, a_2, q)$ to a strip of arbitrary width $R$ in the $y$-direction, and require only that $\| a_2 \|_\infty$ be small relative to $b_-^{1/2}$,
then the {\it asymptotic velocity} associated with energy intervals $\Delta \subset \Delta_{\epsilon,j}$ and the perturbed Hamiltonian $H(A_\epsilon +a,q)$ exists and satisfies the same lower bound as in \eqref{emp1d}. Furthermore, the spectrum in $\Delta$ is absolutely continuous. This means that the edge current is stable with respect to a different class of perturbations than in Theorem \ref{thm-emp1}.

The \emph{asymptotic velocity} associated with a pair of self-adjoint operators $(H_0, H_1)$ is defined in terms of the local wave operators for the pair, see, for example \cite[section 4.5--4.6]{DG}. The local wave operators $\Omega_\pm (\Delta)$ for an energy interval $\Delta \subset \R$ are defined as the strong limits:
\beq\label{eq:local-wo1}
\Omega_\pm (\Delta ) = s-\lim_{t \rightarrow \pm \infty} e^{i t H_1} e^{-i t H_0} \mathbb{P}_{0,ac} (\Delta),
\eeq
where $\mathbb{P}_{0,ac} (\Delta)$ is the spectral projector for the absolutely continuous subspace of $H_0$ associated with the interval $\Delta$.
For any $\varphi$, we define the \emph{asymptotic velocity} $V_y^\pm (\Delta)$ of the state $\varphi$ by
$$
\langle \varphi,V_y^\pm (\Delta) \varphi \rangle  := \langle \varphi, \Omega_\pm (\Delta) v_y \Omega_\pm (\Delta)^* \varphi \rangle .
$$
In the case that $H_0$ commutes with $v_y$, it is easily seen from the definition \eqref{eq:local-wo1} that
$$
\langle \varphi,V_y^\pm (\Delta) \varphi \rangle  := \lim_{t \rightarrow \pm \infty} \langle \varphi, e^{itH_1} \mathbb{P}_{0,ac} (\Delta) v_y
\mathbb{P}_{0,ac} (\Delta) e^{-itH_1}  \varphi \rangle.
$$


Our main result is the existence of the asymptotic velocity in the $y$-direction for the perturbed operators $H(A_\epsilon+a,q)$
of section \ref{sec:pert1} under less stringent conditions on the norm of the perturbation.
We prove that the asymptotic velocity satisfies the lower bound given in \eqref{emp7} provided the perturbations have compact support in the $y$-direction and that $\| a_2 \|_\infty \leq a_* b_-^{1/2}$. For notational simplicity, we write $H_0 = H_\epsilon= H(A_\epsilon)$, for $\epsilon \geq 0$ and $H_1 = H(A_\eps+a,q)$. It follows from Lemma \ref{lm-b1} that the spectrum of $H_0 = H_\epsilon$ is purely absolutely continuous if $\epsilon \geq 0$ is sufficiently small.
We will write $\mathbb{P}_0 (\Delta)$ for $P_{0, ac}(\Delta)$ because of this.

\begin{theorem}\label{th:asymptotic1}
Let $b_-$, $r$, $n$, $j$, $\delta_j$, $c_j$, $\eps_j$ and $\Delta_{\eps,j}$, for some fixed $\eps \in (0,\eps_j)$,  be the same as
in Theorem \ref{pr-b1}. Suppose that the perturbation $a \in {\rm W}^{1,\infty}(\re^2)$ and $q \in {\rm L}^{\infty}(\re^2)$ have their support in the set $\{ (x,y) ~|~
 |y| < R \}$, for some $0 < R < \infty$.
In addition, suppose that
the perturbation $a_2$ satisfies $\| a_2 \|_\infty \leq a_* b_-^{1/2}$, where $a_* < \frac{c_j}{4} \delta_j^3 \left( \frac{r-1}{r^3} \right)$.
Consider any subinterval $\Delta \subset \Delta_{\eps,j}$ with same midpoint. 
Then for any $\varphi \in {\rm Ran} ~ \mathbb{P}_{A_\eps+a,q}(\Delta)$, we have
$$
 \langle \varphi, V_y^\pm (\Delta) \varphi \rangle \geq \frac{c_j}{4} \delta_j^3 \left( \frac{r-1}{r^3} \right) b_-^{1 \slash 2} \| \varphi \|^2 .
$$
 \end{theorem}

Following section 4 of \cite{HS1}, we prove Theorem \ref{th:asymptotic1} by first proving the existence of local wave operators for the pair
$H_0 := H_\epsilon$ and  $H_1 := H (A_\eps+a,q)$, and any interval $\Delta \subset \Delta_{\epsilon,j}$, as in the theorem.

\begin{lemma}\label{lemma:wave-op1}
Let $(H_0,H_1)$ and $\Delta$ be as defined in Theorem \ref{th:asymptotic1}.
The local wave operators for $(H_0, H_1)$ and interval $\Delta$ exist. That is, the strong limits
$$
{\lim}_{t \rightarrow \pm \infty} e^{it H_1} e^{ -itH_0} \mathbb{P}_{0} (\Delta),
$$
exist and define bounded operators $\Omega_\pm (\Delta)$ on ${\rm L}^2 (\R^2)$.
Consequently, the spectrum of $H_1$ is absolutely continuous in $\Delta$.
\end{lemma}

\begin{proof}
1. We use Cook's method and study the local operators defined by
\bea\label{eq:wo2}
\Omega (t, \Delta) - \mathbb{P}_0 (\Delta)  & = & e^{it H_1} e^{ -itH_0} \mathbb{P}_0 (\Delta) - \mathbb{P}_0 (\Delta)\nonumber \\
 &=& \int_0^t ~\frac{d}{ds} e^{is H_1} e^{ -isH_0} \mathbb{P}_0 (\Delta) ~ds \nonumber \\
 &=& i \int_0^t ~ e^{is H_1} W_\epsilon e^{ -isH_0} \mathbb{P}_0 (\Delta) ~ds,
 \eea
 where $W_\epsilon := H_1 - H_0$.
It suffices to prove that for any smooth vector $\varphi$
\beq\label{eq:cook1}
\lim_{t_1, t_2 \rightarrow \infty} \int_{t_1}^{t_2} W_\epsilon e^{ -isH_0} \mathbb{P}_0 (\Delta) \varphi ~ds = 0 ,
\eeq
and similarly for $t_\ell \rightarrow - \infty$, for $\ell = 1,2$.

\noindent
2. We easily calculate a formula for the perturbation:
$$
W_\epsilon = -2 p_x a_1 - 2 (p_y - \beta_\epsilon (x)) a_2 - i (\partial_x a_1) - i (\partial_y a_2) +a_1^2+ a_2^2 + q .
$$
The coefficients of the first-order operator $W_\epsilon$ are supported in $|y| < R$. Let $\chi_R(y) \in C_0^2 (\R)$ be a nonnegative function
equal to one on $[-R,R]$ and supported in $[-2R,2R]$.
We write the integral in \eqref{eq:cook1} as
$$
\int_{t_1}^{t_2} W_\epsilon e^{ -isH_0} \mathbb{P}_0 (\Delta) \varphi ~ds = \int_{t_1}^{t_2} W_\epsilon (H_0+1)^{-1}
 (H_0 +1) \chi_R (y) e^{ -isH_0} \mathbb{P}_0 (\Delta) \varphi ~ds .
$$
It is clear that $\| W_\epsilon (H_0 + 1)^{-1}\| < \infty$.
Furthermore, we compute the commutator:
$$
(H_0  + 1) \chi_R = \chi_R (H_0  + 1) + Q_\epsilon,
$$
and note that the coefficients of $Q_\epsilon$ are supported in $[-2R,-R] \cup [R, 2R]$, where
$$
Q_\epsilon = -2i(p_y - \beta_\epsilon (x)) \chi_R^\prime + \chi_R ''.
$$
As a consequence, we need to estimate integrals of the form for $\ell = 1,2$:
\beq\label{eq:cook2}
{K_{\ell}} ~{\chi}_{2R}(y) \int_{t_1}^{t_2} e^{ -isH_0} \mathbb{P}_0 (\Delta) \varphi ~ds
= K_{\ell} ~{\chi}_{2R}(y) \int_{t_1}^{t_2} ~ds ~\int_{\omega_{\epsilon,j}^{-1}(\Delta)} e^{i (ky-s \omega_{\epsilon,j}(k))} m_\ell (k) \beta_{\epsilon,j} (k) \psi_{\epsilon,j}(x,k) ~dk,
\eeq
where $\beta_{\epsilon,j}(k) = \langle \hat{\varphi}(\cdot,k), \psi_{\epsilon,j}(\cdot,k)  \rangle$, and the operator $K_1 = W_\epsilon (H_0 +1)^{-1}$ with $m_1(k) = 1 + \omega_{\epsilon,j}(k)$, and $K_2 = W_\epsilon (H_0 +1)^{-1} Q_\epsilon$ with $m_1(k) = 1$. Both operators $K_\ell$ are bounded.

\noindent
3. We use the method of stationary phase to evaluate the long time behavior of the integral in \eqref{eq:cook2}. Let $\Phi(k,y,s) := ky - s \omega_{\epsilon, j}(k)$ denote the phase function.
We then have
$$
\partial_k \Phi(k,y,s) := y - s \omega_{\epsilon, j}^\prime (k).
$$
From Theorem \ref{pr-b1} for $H_0 = H_\epsilon$, we know that for $k \in \omega_{\epsilon, j}^{-1}(\Delta)$, the derivative $\omega_{\epsilon, j}^\prime (k)$ is bounded below by the right side of \eqref{eq:current-pert1}. Consequently, we have the lower bound on the derivative of the phase:
$$
| \partial_k \Phi(k,y,s) \tilde{\chi}_{R}(y) | \geq \frac{1}{2} s c_j \delta_j^3 \left( \frac{r-1}{r^3} \right) b_-^{1/2} - 2R.
$$
This is clearly bounded from below by a positive constant for $s$ large enough.
Since $\omega_{\epsilon,j}(k)$ is analytic, we can differentiate the phase any number of times.
As a consequence, for any $N \geq 1$, the integral in \eqref{eq:cook2} may be bounded above as
$$
\left| \int_{t_1}^{t_2} ~ds ~\int_{\omega_{\epsilon,j}^{-1}(\Delta)} e^{i \Phi(k,s,y)} m_\ell (k) \beta_{\epsilon,j} (k) \psi_{\epsilon,j} (x,k) ~dk \right|
\leq C_N \int_{t_1}^{t_2} ~ ~ \frac{ds}{\langle s \rangle^N } ~ \left| \int_{\omega_{\epsilon,j}^{-1}(\Delta)} e^{i \Phi(k,s,y)}
\Sigma_{j,N}(x,k) ~dk \right| ,
$$
where $\Sigma_{j,N}(x,k) \in L^2 (\R \times \omega_{\epsilon,j}^{-1}(\Delta))$. Upon taking $N \geq 2$, the integral vanishes as $t_1, t_2 \rightarrow \infty$. This, along with the boundedness of the operators $K_\ell$, establishes \eqref{eq:cook1}.
\end{proof}

\noindent
The absolutely continuity of the spectrum of $H(A_\epsilon +a,q)$ is the interval $\Delta$ was already established using the Mourre estimate in section \ref{sec:pert1} under smallness bounds on the perturbation $(a,q)$ and $\nabla a$. This lemma establishes the absolute continuity of the spectrum in $\Delta$ under different conditions of the perturbation.
We now prove Theorem \ref{th:asymptotic1}.

\begin{proof}
1. The local wave operators $\Omega_\pm (\Delta)$ satisfy a local intertwining property:
\beq\label{eq:intertwing1}
\Omega_\pm (\Delta)^* \mathbb{P}_1(\Delta) = \mathbb{P}_0 (\Delta) \Omega_\pm (\Delta)^*.
\eeq
Let $\varphi = \mathbb{P}_1 (\Delta) \varphi$ and recall that the velocity operator for $H(A_\epsilon+a,q)$ is given in \eqref{emp1b}.
As a consequence, we compute
\beas
\langle \varphi, V_y^\pm (\Delta) \varphi \rangle & = & \langle \varphi, \Omega_\pm (\Delta) v_{y,A_\epsilon+a,q} {\Omega_\pm} (\Delta)^* \varphi \rangle \nonumber \\
&=& \langle \Omega_\pm (\Delta)^* \mathbb{P}_1 (\Delta) \varphi,  v_{y,A_\epsilon+a,q} {\Omega_\pm} (\Delta)^* \mathbb{P}_1(\Delta)\varphi \rangle \nonumber \\
 & = & \langle  \mathbb{P}_0 (\Delta) \Omega_\pm (\Delta)^* \varphi,  v_{y,A_\epsilon+a,q}
 \mathbb{P}_0 (\Delta) \Omega_\pm (\Delta)^* \varphi \rangle
 \eeas
\noindent
2. We now recall that $v_{y,A_\epsilon+a,q} = v_{y, \epsilon} - a_2$. The operator $\mathbb{P}_0 (\Delta) v_{y,\epsilon} \mathbb{P}_0 (\Delta)$
 is bounded below by $\frac{c_j}{2} \delta_j^3 \left( \frac{r-1}{r^3} \right) b_-^{1 \slash 2}$ according to Theorem \ref{pr-b1}. Our hypotheses on $a_2$ is that $\| a_2 \|_\infty \leq a_* b_-^{1/2} < \frac{c_j}{4} \delta_j^3 \left( \frac{r-1}{r^3} \right) b_-^{1 \slash 2}$. Consequently, we obtain
\beq \label{eq:z}
\langle \varphi, V_y^\pm (\Delta) \varphi \rangle \geq \frac{c_j}{4} \delta_j^3 \left( \frac{r-1}{r^3} \right) b_-^{1 \slash 2}
 \| \mathbb{P}_0 (\Delta) \Omega_\pm (\Delta)^* \varphi \|^2.
\eeq
We again use the intertwining relation \eqref{eq:intertwing1} to write
$$
\| \mathbb{P}_0 (\Delta) \Omega_\pm (\Delta)^* \varphi \| = \| \Omega_\pm (\Delta)^* \mathbb{P}_1 (\Delta) \varphi \| =
\| \varphi \|,
$$
since the local wave operators are local partial isometries. Using this identity in the lower bound \eqref{eq:z} establishes the result.
\end{proof}

For generalized Iwatsuka models, M\v{a}ntoiu and Purice \cite{MP}
 proved minimum and maximum velocity estimates. Their Theorem 5.2 requires that the magnetic field $b(x)$ be positive, bounded $0 < b_- \leq
  b(x) \leq b_+ < \infty$, and that $\lim_{x \rightarrow \pm \infty} b(x) = b_\pm$. For these general models, the band functions may have several critical points that is not the case for the sharp and soft Iwatsuka models studied here.
 In our case, their results apply to the unperturbed operator $H_\epsilon$, for $\epsilon \geq 0$.
For $H_\epsilon$, their main result may be stated as follows.
 Let $\Delta$ be any energy interval as in Theorem \ref{th:asymptotic1}. With respect to $\Delta$, we define
$\rho_\Delta = \inf_{k \in \omega_{\epsilon,j}^{-1}(\Delta)} \omega_{\epsilon,j}^\prime (k)$ and $\theta_\Delta = \sup_{k \in \omega_{\epsilon,j}^{-1}(\Delta)} \omega_{\epsilon,j}^\prime (k)$. These constants play the role of the minimum and maximum velocity, respectively.
Let $F$ be a real-valued function on $\R$ with support outside of the interval $[ \rho_\Delta, \theta_\Delta ]$.
Then, we have
$$
\int_1^\infty ~\frac{dt}{t} ~\|F(|y|/t) e^{-itH_\epsilon} \mathbb{P}_\epsilon (\Delta) \varphi \|^2 \leq C \| \varphi \|^2, ~~\forall \varphi \in {\rm L}^2 (\R^2).
$$
This means that for states with energy in $\Delta$, the time evolution of the $y$-coordinate, $y(t) = e^{itH_\epsilon} y e^{-itH_\epsilon}$ satisfies
$$
\rho_\Delta t \leq |y(t)| \leq \theta_\Delta t, ~~t \rightarrow +\infty,
$$
and a similar bound for $t \rightarrow - \infty$. The constants $\rho_\Delta$ and $\theta_\Delta$ are thus the  minimum and maximum velocities
associated with $\Delta$.  In our situation the edge current is also strongly localized to the strip $|x| \leq C b_-^{-1/2}$.


\end{document}